\documentclass[12pt, draftclsnofoot, onecolumn]{IEEEtran}

\ifCLASSINFOpdf
\else
\fi
\usepackage{amsmath,epsfig,comment}
\usepackage{amsthm}
\usepackage{tcolorbox}
\usepackage{amssymb}
\usepackage{mathrsfs}
\usepackage{amsmath}
\usepackage{tikz}
\usetikzlibrary{arrows,decorations.markings,automata}
\usepackage{hyperref}
\hypersetup{
	colorlinks=true,
	linkcolor=black,
	filecolor=magenta,
	urlcolor=cyan,
	citecolor=blue
}

\usepackage{cases} 
\usepackage{amssymb,amsmath,cite}
\usepackage{epsfig}
\usepackage{color}
\usepackage{bm}
\usepackage[ruled]{algorithm2e}
\usepackage{graphicx}
\usepackage{subfigure}
\usepackage{algorithmic}
\usepackage{multirow}
\usepackage{soul}
\usepackage{extarrows}

\newtheorem{theorem}{Theorem}

\newtheorem{example}{Example}
\newtheorem{assumption}{Assumption}

\newtheorem{lemma}{Lemma}

\usepackage{float}
\usepackage{booktabs}
\usepackage{multirow}
\usepackage{soul}
\def\blue{\color{black}}

\DeclareMathAlphabet\mathbfcal{OMS}{cmsy}{b}{n}

\newcommand{\mat}[1]{\mathbf{#1}}
\def\s{\mathbf{s}}

\def\I{\mathcal{I}}
\def\RR{\mathcal{R}}
\def\A{\mathbf{A}}

\def\n{\mathbf{n}}
\def\d{\mathbf{d}}

\def\I{\mathcal{I}}
\def\R{\mathbb{R}}
\def\C{\mathbb{C}}
\def\H{\mathbf{H}}

\def\Re{\mathcal{R}}
\def\Im{\mathcal{I}}
\def\y{\boldsymbol{y}}

\def\z{\mathbf{z}}
\def\x{\boldsymbol{x}}

\def\h{\mathbf{h}}
\def\u{\mathbf{u}}
\def\v{\mathbf{v}}

\def\a{\boldsymbol{\Lambda}}
\newtheoremstyle{noparens}%
{}{}%
{\itshape}{}%
{\bfseries}{.}%
{ }%
{\thmname{#1}\thmnumber{ #2}\mdseries\thmnote{ #3}}

\theoremstyle{noparens}

\begin{document}
%
\title{Efficient CI-Based One-Bit Precoding for Multiuser Downlink Massive MIMO Systems with PSK Modulation}
%
%
%

  \author{\IEEEauthorblockN{Zheyu Wu, Bo Jiang, Ya-Feng Liu, Mingjie Shao, and Yu-Hong Dai}
	\thanks{
	 Part of this work has been presented at the IEEE International Conference on Acoustics, Speech, and Signal Processing (ICASSP)  2022 \cite{conference}.
	
	     		Z. Wu, Y.-F. Liu, {and Y.-H. Dai} are with the State Key Laboratory of Scientific and Engineering Computing, Institute of Computational Mathematics and Scientific/Engineering Computing, Academy of Mathematics and Systems Science, Chinese Academy of Sciences, Beijing 100190, China (e-mail: \{wuzy, yafliu, dyh\}@lsec.cc.ac.cn). B. Jiang is with the Ministry of Education Key Laboratory for NSLSCS, School of Mathematical Sciences, Nanjing Normal University, Nanjing 210023, China (e-mail: jiangbo@njnu.edu.cn). {\color{black}M. Shao is with the School of Information Science and Engineering, Shandong University, Qingdao, China (email: mingjieshao@sdu.edu.cn)}
			}
}

\maketitle
\begin{abstract}
In this paper, we consider the one-bit precoding problem for the multiuser downlink massive multiple-input multiple-output (MIMO) system with phase shift keying (PSK) modulation. {\color{black}We} focus on the celebrated constructive interference (CI)-based problem formulation. We first establish the NP-hardness of the problem (even in the single-user case), which reveals the intrinsic difficulty of globally solving the problem. Then, we propose a novel negative $\ell_1$ penalty  model for the considered problem, which penalizes the one-bit constraint into the objective {\color{black}by} a negative $\ell_1$-norm term, and show the equivalence between (global and local) solutions of the original problem and the penalty problem when the penalty parameter is sufficiently large.  We further transform the penalty model into an equivalent min-max problem and propose an efficient alternating proximal/projection gradient descent ascent (APGDA) algorithm for solving it, which performs a proximal gradient decent over one block of variables and a projection gradient ascent over the other block of variables alternately.  The APGDA algorithm enjoys a low per-iteration complexity and is guaranteed to converge to a stationary point of the  min-max problem and a local minimizer of the penalty problem. To further reduce the computational cost, we also propose a low-complexity implementation of the APGDA algorithm, where the values of the variables will be fixed in later iterations once they satisfy the one-bit constraint. Numerical results show that, compared {\color{black}to} the state-of-the-art CI-based algorithms, both of the  proposed algorithms generally achieve better bit-error-rate (BER) performance with lower computational cost.

\end{abstract}

\begin{IEEEkeywords}
\vspace{-0.2cm}
Constructive interference, massive MIMO, min-max problem, negative $\ell_1$ penalty, one-bit precoding.
\end{IEEEkeywords}

%
\IEEEpeerreviewmaketitle

\section{Introduction}\label{sec:introduction}
\IEEEPARstart{M}{assive} multiple-input multiple-output (MIMO), which {\blue may} deploy hundreds of antennas at the base station (BS),  is a key and effective technology for significantly improving the {\color{black}spectral} and energy efficiency of 5G and beyond wireless communication systems \cite{massivemimo2,massivemimo1,massivemimo3}.
{\blue However, since the  numbers of radio-frequency (RF) chains and analog-to-digital converters (ADCs)/digital-to-analog converters (DACs) scale up by the number of antennas, massive MIMO systems face  great challenges of high hardware complexity and power consumption.
To alleviate the above issues, researchers have exploited two research directions: cutting down the number of RF chains and ADCs/DACs that drive the large-antenna array, and reducing the resolution/quality of each ADC/DAC and RF chain.
In the downlink, the former paradigm leads to hybrid analog digital (AD) precoding \cite{AD1,AD2,AD3}, while the latter corresponds to precoding under low-resolution DACs \cite{SQUID}.
In particular, one-bit precoding, where the BS is employed with the cheapest and most power-efficient one-bit DACs, has attracted a lot of recent research interests \cite{ADC4,ADC5,linear2,linear1,SQUID,C3PO2,TITpaper,MAGIQ,IDE,ADMM,CIfirst,QCECI,CImodel,PBB,GMSM,CIBB,CI1bitoverview,sep2,sep3,GEMM,seppsk}.
On the one hand, the one-bit signal has low peak-to-average power ratio (PAPR), which is also favorable to save energy for the accompanying power amplifies (PAs) in the RF chains.
On the other hand, the use of one-bit DACs imposes stringent discrete signal constraints on the transmitted signal, which renders one-bit precoding a challenging problem to solve.
In this paper, we {\color{black}focus} on the one-bit precoding problem in the multiuser downlink massive MIMO system with phase shift keying (PSK) modulation.
We contribute to the fundamental computational complexity  analysis and   efficient algorithm designs for   the one-bit precoding problem.
}

\subsection{Related Works}\label{relatedworks}

Early works on downlink transmission with one-bit DACs have mainly focused on {\color{black}the analysis and design} of linear-quantized precoding schemes, in which the precoders are obtained by simply quantizing the classical linear precoders \cite{SQUID}\cite{linear2,linear1,TITpaper}. {\color{black}In particular, the authors in \cite{TITpaper} have given a unified performance analysis for a wide class of linear-quantized precoding schemes, including quantized matched-filter (MF) and zero-forcing (ZF) precoding as special cases. They have also identified the optimal linear-quantized precoder within the considered class.}
 Despite the advantage of their low computational complexities, {\color{black} the linear-quantized} precoders usually suffer from severe bit error rate (BER) floors, especially in the high signal-to-noise ratio (SNR) regime.

{\blue
The nonlinear precoding schemes usually take symbol-level metrics for optimization, among which  minimum mean square error (MMSE) \cite{SQUID}\cite{C3PO2,MAGIQ,IDE,ADMM}, constructive inference (CI) \cite{CIfirst,QCECI,CImodel,PBB,GMSM,CIBB,CI1bitoverview}, and symbol error probability (SEP) \cite{sep2,sep3,GEMM,seppsk} are probably the most widely considered performance metrics.
MMSE measures the average distance between the received signal and its corresponding constellation symbol, which is a general metric that can be applied to both QAM and PSK constellations.
Several one-bit precoders that optimize the MMSE  have been proposed in \cite{SQUID}, including the semidefinite relaxation (SDR) precoder and the more computationally efficient squared-infinity norm Douglas-Rachford splitting (SQUID) precoder. To further reduce the computational cost of SQUID, {\color{black}the authors in \cite{C3PO2} proposed two other precoders, called C1PO and C2PO, based on the biconvex relaxation technique.
In \cite{MAGIQ}, the authors proposed a greedy iterative precoder named MAGIQ, which exhibits  slightly better  performance compared to that of the SQUID, C1PO, and C2PO algorithms. 
Another precoder based on a modified alternating direction method of multiplier (ADMM)  framework, referred to as the iterative discrete estimation (IDE), was introduced in \cite{IDE}.  In the same work, the authors also proposed an efficient low-complexity implementation of IDE, named IDE2, which achieves similar performance to IDE but with a significantly reduced computational cost. }
More recently, one-bit precoding with SEP minimization has received increasing attentions \cite{sep2,sep3,GEMM,seppsk}.
The crux lies in that the SEP expression for PSK constellation involves integral and does not admit a closed form \cite{digitalcommunication}, which greatly increases the difficulty of precoding algorithm designs.
}

{\blue
For PSK constellation, the most widely-adopted design metric is CI. The concept of CI exploits the sectorial decision region property for PSK constellation.
Intuitively speaking, it tries to alter the multiuser interference to be constructive to push the noise-free received signal deep into the decision region and away from the decision boundary.
{\color{black}Compared to the MMSE metric, the CI metric exploits the beneficial interference inherent in multiuser transmission, thereby more effectively leveraging the advantages of symbol-level precoding \cite{CI1bitoverview}\cite{CItutorial}}; compared to the SEP {\color{black}metric}, the CI metric is much easier to optimize.
In fact, the CI criterion is closely related to the SEP criterion. In particular, it has been shown in \cite{sep3} and \cite{diversity} that for PSK constellation, maximizing the CI effect can be seen as minimizing an approximation of the SEP.
The concept of CI has been well {\color{black}studied for symbol-level precoding without one-bit constraint \cite{CItutorial}\cite{CI1,CI2,CI3}; more recently, it has been exploited for one-bit precoding \cite{CIfirst,QCECI,CImodel,PBB,GMSM,CIBB,CI1bitoverview}.
}

{\color{black}To the best of our knowledge, the first to incorporate the idea of CI into one-bit precoding design is \cite{CIfirst}\cite{QCECI}. Subsequently, the authors in \cite{CImodel} proposed an alternative CI-based model, known as the  symbol scaling model, which admits a simpler formulation and has been shown in  \cite{CIequivalent} to be equivalent to the model in \cite{CIfirst} \cite{QCECI}.}  {\color{black}By far, various} algorithms have  been proposed for solving the CI-based {\blue one-bit precoding problem}, most of which are based on the linear programming (LP) relaxation model \cite{CIfirst,QCECI, CImodel, PBB,GMSM}. 
Generally speaking, this kind of approach consists of two stages: in the first stage, the LP relaxation model is solved; in the second stage, some optimization or greedy techniques are utilized to determine the values of elements of the LP relaxation solution that do not satisfy the one-bit constraint. The theoretical support of this kind of approach is that most entries of  the solution of the LP relaxation already satisfy the one-bit constraint, and thus in the second stage only a subproblem with a reduced dimension needs to be considered \cite{PBB}. Different techniques in the second stage lead to different algorithms.  Specifically, the {\color{black}maximum safety margin} (MSM) algorithm \cite{CIfirst,QCECI,CImodel} directly quantizes the LP relaxation  solution to satisfy the one-bit constraint\footnote{{\color{black}The term ``MSM'' was introduced in \cite{CIfirst}  \cite{QCECI} to refer to the specific algorithm. In \cite{CImodel}, the corresponding algorithm is called ``constructive'' in its simulations. However, to maintain simplicity and acknowledge the equivalence of the algorithms presented in \cite{CIfirst, QCECI, CImodel}, we uniformly refer to them as ``MSM"  throughout the paper.}}, which is the most straightforward way to obtain a one-bit solution from the solution of the LP relaxation.  However, its BER performance is often unsatisfactory. The {\color{black}partial branch-and-bound} (P-BB) algorithm \cite{PBB} solves the subproblem in the second stage to global optimality with an elaborately designed branch-and-bound procedure, which  achieves the best performance of this kind of approaches but is unsuitable for practical implementation due to its high computational complexity. {\color{black} To make a balance between complexity and performance,  some algorithms employ the greedy technique in the second stage to determine the values of the non-one-bit elements of the LP relaxation solution in some custom-designed sequential manners. Examples of such algorithms include the ordered partial sequential update (OPSU) algorithm \cite{PBB} and the greedy MSM (GMSM) algorithm \cite{GMSM}. 
These greedy approaches can greatly enhance the performance of the MSM algorithm {\color{black}with slightly increased} computational cost, making them considerably more efficient than the P-BB algorithm.

 }
{\color{black} We summarize the existing models and/or algorithms for one-bit precoding under PSK constellations in Table \ref{Compared algorithms}. Few methods can handle difficult cases (e.g., large number of users, high-order PSK signals) and maintain a low computational complexity at the same time. 
}

\begin{table}
\caption{A summary of models and algorithms for one-bit precoding with PSK modulation.}
\label{Compared algorithms}
\centering
\fontsize{8}{8}\selectfont
\begin{tabular}{|c|c|c|c|c|}
\hline
Algorithm& Design principle &Optimization model and/or technique&Complexity&Error rate\\
\hline
\multirow{2}{*}{ZF \cite{SQUID}}&\multirow{2}{*}{Block-level ZF}&\multirow{4}{*}{Direct quantization}&\multirow{4}{*}{Low}&\multirow{4}{*}{Poor}\\
&&&&\\
\cline{1-1}\cline{2-2}
\multirow{2}{*}{WF\cite{SQUID}}&\multirow{2}{*}{Block-level MMSE}&&&\\
&&&&\\
\hline
\multirow{2}{*}{C1PO \cite{C3PO2}}&&{Biconvex relaxation and}&\multirow{10}{*}{Moderate}&\\
&\multirow{4}{*}{Symbol-level}&alternating minimization&&\multirow{8}{*}{{\color{black}Satisfactory} in easy cases}\\
\cline{1-1} \cline{3-3}
\multirow{2}{*}{C2PO \cite{C3PO2}}&&{Biconvex relaxation and}&&\\
&& forward-backward splitting&&\\
\cline{1-1}\cline{3-3}
\multirow{2}{*}{SQUID \cite{SQUID}}&\multirow{3}{*} {MMSE} &\multirow{2}{*}{Douglas-Rachford splitting}&&\\
&&&&\\
\cline{1-1}\cline{3-3}
\multirow{2}{*}{{\color{black}MAGIQ \cite{MAGIQ}}}&&\multirow{2}{*}{{\color{black}Greedy technique}}&&{but failed in difficult cases}\\
&&&&\\
\cline{1-1}\cline{3-3}
\multirow{2}{*}{{\color{black}IDE2 \cite{IDE}}}&&\multirow{2}{*}{{\color{black}ADMM}}&&\\
&&&&\\
\cline{1-1}\cline{2-2}\cline{3-3}\cline{4-4}
\multirow{2}{*}{MSM\cite{CIfirst,QCECI,CImodel}}&&\multirow{2}{*}{LP relaxation and direct quantization}&&\\
&&&Moderate, {\color{black} generally} higher&\\
\cline{1-1}\cline{3-3}\cline{5-5}
\multirow{2}{*}{OPSU\cite{PBB}, GMSM\cite{GMSM}}&&\multirow{2}{*}{LP relaxation and greedy search}&than MMSE methods&{{\color{black}Good} in general but}\\
&\multirow{2}{*}{Symbol-level}&&&degraded in difficult cases\\
\cline{1-1}\cline{3-3}\cline{4-4}\cline{5-5}
\multirow{2}{*}{P-BB\cite{PBB}}&&\multirow{2}{*}{LP relaxation and branch-and-bound}&{Prohibitively high}&\multirow{6}{*}{{\color{black}Good}}\\
&&&in large systems&\\
\cline{1-1}\cline{3-3}\cline{4-4}
\multirow{2}{*}{NL1P (this paper)}&{CI}&\multirow{2}{*}{Negative $\ell_1$ penalty reformulation }&&\\
&&&Moderate, {\color{black}generally} higher &\\
\cline{1-1}
\multirow{2}{*}{ANL1P (this paper)}&&and min-max optimization& than MMSE methods&\\
&&&&\\
\hline
\end{tabular}
\end{table}

\subsection{Our Contributions}
{This paper considers the one-bit precoding design problem for massive MIMO systems with PSK modulation and focuses on the CI-based symbol scaling model \cite{CImodel}.
 The main contribution of this paper is  an efficient negative $\ell_1$ penalty (NL1P) approach for solving large-scale {\color{black} one-bit precoding} problems arising from the massive MIMO scenario. Two key features of the proposed approach are as follows. First, our approach is based on a novel penalty model, which is shown to be equivalent to the original problem both globally and locally when the penalty parameter is sufficiently large. This is in sharp contrast to the LP relaxation model considered in the previous works \cite{CIfirst,QCECI,CImodel,PBB,GMSM}, {\color{black} whose optimal solution is definitely different from that of the original problem}. Second, the dominant cost of the proposed approach at each iteration is two matrix-vector multiplications and one projection onto the simplex, which makes it particularly suitable for solving large-scale one-bit precoding problems {\color{black}in the massive MIMO systems.}

We summarize the contributions of the paper as follows.
\begin{enumerate}
\item \emph{Complexity Analysis:} We characterize the complexity status of the considered one-bit precoding problem. Specifically, we show that the considered problem is NP-hard even in the single-user case and strongly NP-hard in the general case. The complexity results fill a theoretical gap, as the complexity status of the problem remains unknown (in spite of the existence of various heuristic approaches for solving the problem).
\item \emph{Novel Penalty Model}: We propose a novel negative $\ell_1$ penalty model for the considered problem, in which the one-bit constraint is penalized into the objective with a negative $\ell_1$-norm term.  We show that when the penalty parameter is sufficiently large, the penalty model is  an exact reformulation of the original problem, in the sense that the two problems share the same global and local solutions. {\color{black}The proposed penalty model is the first continuous reformulation of the original discrete model and is more favorable for the algorithmic design (compared to the discrete model). }

\item \emph{Efficient Algorithms}: To solve the penalty model, we further transform it into an equivalent min-max problem. We propose an efficient alternating proximal/projection gradient descent ascent (APGDA) algorithm for solving a class of non-smooth nonconvex-concave min-max problems (which includes our problem as a special case) and prove its convergence.
    {\blue More specifically}, we show that the APGDA algorithm is guaranteed to converge to a stationary point of the min-max problem and a local minimizer of the penalty problem.  We also propose a low-complexity implementation of the proposed APGDA algorithm when applied to solve our interested penalty problem. Simulation results show that both the proposed algorithm and its low-complexity implementation  generally outperform the state-of-the-art CI-based algorithms in terms of both the BER performance and the computational efficiency.

    \end{enumerate}

{\color{black}Compared to its conference version \cite{conference}, this paper has made significant progress in theory, algorithm, and  numerical simulation. 
First, this paper gives a  complexity analysis of the one-bit precoding problem, which was not covered in \cite{conference}.
 Second, the APGDA algorithm proposed in this paper is a substantial generalization of the alternating optimization (AO) algorithm proposed in  \cite{conference}. While the AO algorithm {\color{black}in \cite{conference}} was limited to a specific problem, our generalization makes the APGDA algorithm applicable to a broader range of potential applications \cite{ICASSPQCE}.  Moreover, this paper introduces a novel low-complexity implementation of the APGDA algorithm applied to the one-bit precoding problem and obtains an accelerated NL1P approach, which can achieve comparable performance to the NL1P approach in \cite{conference} {\color{black}with reduced} computational cost.
  Lastly, this paper provides more comprehensive numerical results to demonstrate the superiority of our proposed approaches. }

\vspace{-0.2cm}
\subsection{Organization and Notations}\label{notions}
The remaining parts of the paper are organized as follows. Section \ref{background} introduces the system model and the CI-based symbol scaling model for one-bit precoding design. Section \ref{complexity} establishes the complexity status of the considered problem.  A framework of the proposed negative $\ell_1$ penalty approach is given in Section \ref{penaltymodel}, after which an efficient algorithm for solving the penalty model is developed in Section \ref{algorithm}. Simulation results are shown in Section \ref{simulation} and the paper is concluded in Section VII.

Throughout this paper, we use $\R$ and $\C$ to represent the real and complex space, respectively. We use $x$, $\x$, $\mat{X}$, and $\mathcal{X}$ to denote scalar, column vector, matrix, and set, respectively. The symbols $\mathbf{0}$ and $\mathbf{1}$ are column vectors  with all elements being $0$ and $1$, respectively. For a vector $\x$, $\x(i)$ refers to its $i$-th entry, where $x_i$ is also used if it does not cause any ambiguity; $\x\geq\mathbf{0}~(\x>\mathbf{0})$ means that each element of $\x$ is nonnegative (positive). For a matrix $\mat{X}$, $X_{ij}$ refers to its $(i,j)$-th element; $\text{mean}(\mat{X})$ returns the {\color{black}average value} of $\{X_{ij}\}$. For a set $\mathcal{X}$,  $\text{Proj}_\mathcal{X}$ is the projection operator onto set $\mathcal{X}.$  $\text{sgn}(\cdot)$ represents the sign of a real number, which returns $1$ if the number is nonnegative and returns $-1$ otherwise. $\|\cdot\|_p$ denotes the $\ell_p$ norm of the corresponding matrix or vector, where $p\in\{1,2,\infty\}$. $(\cdot)^\mathsf{T}$, $\RR(\cdot)$, $\I(\cdot)$, and $|\cdot|$ return the transpose, the real part, the imaginary part, and the modular of their corresponding argument, respectively. The subdifferential of a convex function $f$ is denoted by $\partial f(\cdot)$.  
$\mathcal{C}\mathcal{N}(\mathbf{0},\sigma^2\mathbf{I})$ represents the zero-mean circularly symmetric complex Gaussian distribution with covariance matrix $\sigma^2\mathbf{I}$, where $\mathbf{I}$ denotes the identity matrix. $B(\x_0,d)$ refers to the ball centered at $\x_0$ with radius $d$, i.e., $B(\x_0,d)=\left\{\x\mid\|\x-\x_0\|_2\leq d\right\}$. $\mathbb{P}(\cdot)$ denotes the probability of the corresponding event. Finally, $j\triangleq\sqrt{-1}$ is the imaginary unit.

\section{Problem Formulation}\label{background}
In this section, we present the problem formulation, including the system model and the CI-based symbol scaling model for one-bit precoding design.
\subsection{System Model}
We consider a {\blue standard flat-fading} downlink multiuser massive MIMO system, in which a BS equipped with $N_t$ antennas serves $K$ single-antenna users simultaneously, where $K\ll N_t$.
{\blue The downlink transmission can be modeled by
$$\y=\H \x_T+\n,$$
where $\x_T\in \C^{N_t}$ is the transmitted signal at the BS, $\H=[\h_1,\dots,\h_K]^\mathsf{T}\in \C^{K\times N_t}$ is the downlink channel from the BS to the users, $\n\sim\mathcal{C}\mathcal{N}(\mathbf{0},\sigma^2\mathbf{I})$ is the additive complex Gaussian noise, and  $\y\in\C^{K\times 1}$  is the received signal vector with $y_k$ representing the received signal of user $k$.
}

{\blue The BS is equipped with one-bit DACs, which enforces that each entry of the transmitted signal $\boldsymbol{x}_T$  can only be chosen from four values. }
Specifically, $\x_T\in\mathcal{X}^{N_t}$, where
$\mathcal{X}=\left\{\pm\frac{1}{\sqrt{2N_t}}\pm\frac{1}{\sqrt{2N_t}}j\right\}.$ Here we normalize $\x_T$ {\blue to unit transmission power}{\color{black}, i.e.,} $\|\x_T\|_2^2=1${\color{black},} for simplicity.
Let $\s=[s_1,\dots,s_K]^\mathsf{T}$ be the intended data symbol vector for the users whose entries are drawn from a unit-norm $M$-PSK constellation, {\blue i.e., $s_k\in\mathcal{S}_M:=\{e^{j\frac{2\pi m}{M}}\mid m=0,\dots,M-1\}$.}
At the receiver side, {\blue  each user  detects the symbol $\hat{s}_k $ from the received signal $y_k$ via nearest-neighbor decoding, i.e., $\hat{s}_k  = \arg\min_{s\in \mathcal{S}_M}\| y_k - s \|_2$.}

{\blue Assuming that the downlink channel $\H$ is known at the BS \cite{linear1,linear2,C3PO2,MAGIQ,TITpaper,IDE,ADMM,CIfirst,QCECI,CImodel,GMSM,PBB,CIBB,CI1bitoverview,GEMM,sep2,sep3}, our goal is to design the transmitted signal $\x_T$ such that the SEP,   defined as  $\mathbb{P}(\s\neq\hat{\s})$, is as low as possible. }
In this paper, we focus on the CI formulation {\blue for the one-bit precoding} problem \cite{PBB,CImodel,CIBB,GMSM,CIfirst,QCECI,CI1bitoverview}.
{\blue CI is a symbol-level precoding scheme, where the transmitted signal $\x_T$ is optimized for each realization of the pair of channel and symbol $(\H,\s)$.}

\subsection{CI-Based Symbol Scaling Model for One-Bit Precoding}
CI refers to {\color{black} the} interference that pushes the received signal away from {\blue  their corresponding decision boundaries} of the modulated-symbol constellation, which thus contributes to the useful signal power \cite{CItutorial}. {\color{black}In the context of symbol-level precoding design, the CI metric exploits the idea of CI, aiming to maximize the distance between the noise-free received signal and the decision boundary of the desired symbol. By doing so, the true received signal (with the additive Gaussian noise added) is less likely to fall outside the decision region, thus achieving a lower SEP. In comparison, the MMSE metric seeks to suppress all the interference so that the noise-free received signal is as close to the intended symbol as possible. Consequently, the MMSE metric fails to exploit the finite alphabet nature of the constellation symbols and the special shape of their decision regions, {\color{black}and is generally inferior to the CI metric in terms of minimizing the system SEP, which is an ideal performance metric in the context of symbol-level precoding.}
Please see \cite{CI1bitoverview} and \cite{CItutorial}  for more detailed discussions on CI and its comparison with the MMSE metric. 
In this paper, we adopt the CI metric as the performance metric for one-bit precoding design.}  In this subsection, we  introduce the mathematical formulation of the CI metric and the corresponding symbol scaling model proposed in \cite{CImodel}.

\begin{figure}
\centering
\includegraphics[scale=0.33]{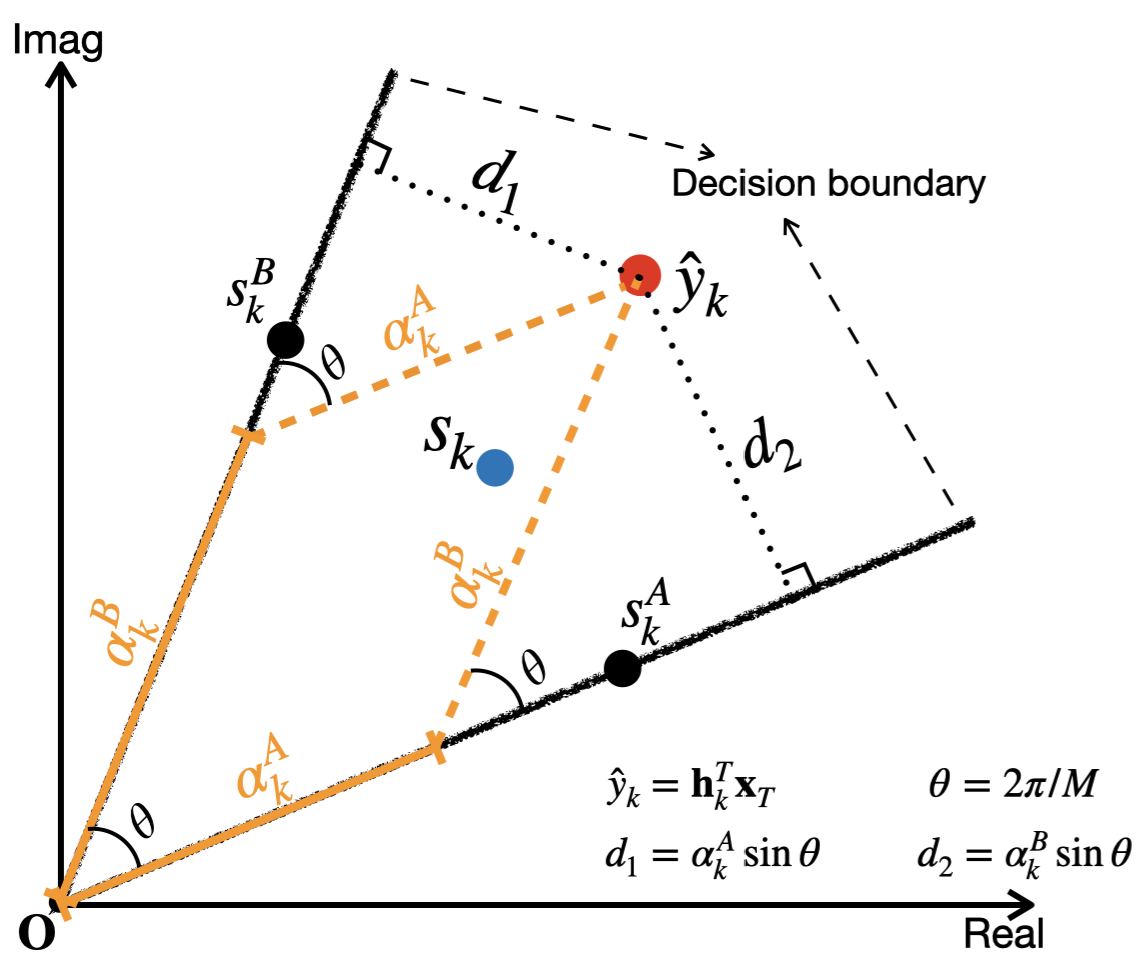}
\caption{An illustration of the CI formulation for 8-PSK.}
\label{fig3}
\vspace{-0.5cm}
\end{figure}
For  clarity, in Fig. \ref{fig3} we depict a piece of the decision region for 8-PSK modulation, where without loss of generality we assume that the data symbol for user $k$ is $s_k=e^{j\frac{2\pi}{M}}$.
{\color{black}The goal of the CI metric is to maximize the distance between the noise-free received signal, denoted by $\hat{y}_k$ in Fig. \ref{fig3},  and the decision boundary of $s_k$, i.e., maximizing $\min\{d_1,d_2\}$. To formulate $d_1$ and $d_2$ mathematically, we decompose $\hat{y}_k$ along  the directions of the two decision boundaries as
$$\hat{y}_k=\h_k^\mathsf{T}\x_T=\alpha_k^As_k^A+\alpha_k^Bs_k^B,$$
where $s_k^A$ and $s_k^B$ are the unit vectors in the directions of the decision boundaries given by $s_k^A=s_ke^{-j\frac{\pi}{M}}~~\text{and}~~s_k^B=s_ke^{j\frac{\pi}{M}},$ respectively. 
Then, as can be observed from Fig. \ref{fig3}, we have
\begin{equation*}
\begin{aligned}
\min\left\{d_1,d_2\right\}=\min\left\{\alpha_k^A\sin\theta,\alpha_k^B\sin\theta\right\}=\min\left\{\alpha_k^A,\alpha_k^B\right\}\sin\theta.
\end{aligned}
\end{equation*}}
\hspace{-0.2cm}Since $\theta=\frac{2\pi}{M}$ is known as long as  the constellation level $M$ is given, the distance is only determined by $\min\{\alpha_k^A,\alpha_k^B\}$.

Based on the above discussions, the CI effect for all users in the system can be characterized by the value of  $\min_{k\in\{1,2,\dots, K\}}\{\alpha_k^A, \alpha_k^B\}$. 
  Accordingly, the one-bit precoding design problem that maximizes the CI effect can be formulated as
\begin{subequations}
\begin{align}
\max_{\x_T}~&\min_{k\in\{1,2,\dots, K\}}~\left\{\alpha_k^A,\alpha_k^B\right\}\notag\\
\text{s.t. }~&\h_k^\mathsf{T}\x_T=\alpha_k^As_k^A+\alpha_k^Bs_k^B,\quad k=1, 2,\dots,K,\tag{1}\\
~&\x_T(i)\in\left\{\pm\frac{1}{\sqrt{2N_t}}\pm\frac{1}{\sqrt{2N_t}}j\right\},\quad i=1, 2,\dots, N_t.\notag
\end{align}
\end{subequations}
By denoting $\hat{\x}_T=\sqrt{2N_t}\x_T$ and $\hat{\H}=\frac{\H}{\sqrt{2N_t}}$, we can further remove the problem-dependent quantity $\frac{1}{\sqrt{2N_t}}$ from the constraint on $\x_T$. With a little bit notational ambiguity, we still use $\x_T$ and $\H$, then problem (1) can be rewritten as
\begin{subequations}
\begin{align}
\max_{\x_T}~&\min_{k\in\{1,2,\dots,K\}}~\left\{\alpha_k^A,\alpha_k^B\right\}\notag\\
\hspace{-0.2cm}\text{(}\text{P}_0\text{)}\hspace{0.9cm}\text{s.t. }~&\h_k^\mathsf{T}\x_T=\alpha_k^As_k^A+\alpha_k^Bs_k^B,\quad k=1, 2,\dots,K,\label{2a}\\
~&\x_T(i)\in\left\{\pm1 \pm j\right\}, \quad i=1, 2,\dots, N_t.
\end{align}
\end{subequations}
We refer to (P$_0$) as the CI-based symbol scaling model for one-bit precoding design, {\blue which is the main optimization problem of this paper}.

\section{Complexity Analysis}\label{complexity}

{\blue Problem (P$_0$) is a large-scale integer program, as the number of antennas $N_t$ is usually large in massive MIMO systems.
Such a problem is considered challenging to solve from the optimization viewpoint.}
In spite of the existence of various {\blue algorithmic} works on problem (P$_0$), its complexity status, {\blue or how intrinsically difficult it is from the perspective of theoretical computer science,} remains unknown {\blue in the literature}.
In this section, we fill this theoretical gap, i.e., characterizing the complexity of problem (P$_0$).

We first consider the case where there is only a single user in the system.
\begin{theorem}\label{th1}
The CI-based one-bit precoding problem (P$_0$) is NP-hard in the single-user case, i.e., $K=1$.
\end{theorem}
\begin{proof}
See Appendix \ref{AppendixA}.
\end{proof}

The following Theorem \ref{th2} considers the more general  multiuser case with $K\geq 1$.
\begin{theorem}\label{th2}
The CI-based one-bit precoding problem (P$_0$) is strongly NP-hard. Moreover, there is no polynomial-time constant approximation algorithm for (P$_0$), unless P~$=$~NP.
\end{theorem}
\begin{proof}
See Appendix \ref{AppendixB}.
\end{proof}

The above complexity results reveal that the (worst-case) computational complexity of globally solving (P$_0$) is exponential (if P~$\neq$~NP), which is prohibitively high for the massive MIMO system whose corresponding problem size is large. In addition, since the precoding scheme {\blue is} symbol-level based, (P$_0$) must be solved at the symbol rate, which imposes high requirement on the efficiency of the corresponding algorithm.
As such, instead of insisting on finding the optimal solution, we focus our attention on designing efficient algorithms for finding high-quality solutions {\blue to} problem (P$_0$).

\section{Proposed Negative $\ell_1$ Penalty Approach} \label{penaltymodel}

In this section, we first introduce a compact form of problem (P$_0$), which is more favorable for the algorithmic design. Then, we transform the compact form into a novel negative $\ell_1$ penalty model and give the algorithmic framework of the proposed negative $\ell_1$ penalty approach.
\subsection{A Compact Form of (P$_0$)}
In this subsection, we briefly introduce a compact form of (P$_0$) proposed in \cite{CImodel}. Recall that $\alpha_k^A$ and $\alpha_k^B$ are both real numbers.
Therefore, by rewriting the complex-valued constraints (\ref{2a}) into the real-valued form, we can express $[\alpha_k^A, \alpha_k^B]^\mathsf{T}$ explicitly as a function of $\h_k$, $s_k$, and $\x=[\RR(\x_T)^\mathsf{T},\I(\x_T)^\mathsf{T}]^\mathsf{T}$. Moreover, the original maximization problem can be converted into a minimization problem (by adding a negative sign in the objective). Then we arrive at the following compact form:
\begin{equation}\label{maxminform}
\begin{aligned}
\min_{\x}~&\max_{l\in\{1,2,\dots,2K\}}\alpha_l\\
\qquad\text{s.t. }~&\a=\mathbf{A}\x,\\
~&x_i\in\left\{-1,1\right\},~i=1,2,\dots, 2N_t,
\end{aligned}
\end{equation} where $\a\hspace{-0.05cm}=\hspace{-0.05cm}-\hspace{-0.05cm}
\left[\alpha_1^A, \alpha_1^B, \dots, \alpha_K^A, \alpha_K^B\right]^\mathsf{T}\hspace{-0.07cm}\triangleq\hspace{-0.07cm}[\alpha_1,\alpha_2,\dots,\alpha_{2K}]^\mathsf{T}\hspace{-0.05cm}\in\R^{2K}$ and $\A\hspace{-0.05cm}=\hspace{-0.05cm}-\hspace{-0.05cm}\left[\mathbf{V}_1^\mathsf{T},\mathbf{V}_2^\mathsf{T},\dots,\mathbf{V}_K^\mathsf{T}\right]^\mathsf{T}\hspace{-0.1cm}\in\R^{2K\times 2N_t}$ with
\begin{equation*}\label{Vk}
\mathbf{V}_k=
\displaystyle\frac{
\left[
\begin{matrix}
\I(s_k^B)&-\RR(s_k^B)\\-\I(s_k^A)&\RR(s_k^A)
\end{matrix}\right]
\left[
\begin{matrix}
\RR(\h_k^\mathsf{T})&-\I(\h_k^\mathsf{T})\\\I(\h_k^\mathsf{T})&\RR(\h_k^\mathsf{T})
\end{matrix}\right]}{\RR(s_k^A)\I(s_k^B)-\I(s_k^A)\RR(s_k^B)}.
\end{equation*}
See \cite{CImodel} for detailed derivations.

The constraint $\a=\mathbf{A}\x$ in problem \eqref{maxminform} can be further substituted into the objective, which leads to the following  form:
\begin{equation*}\label{Pe}
\text{(P)}\qquad\min_{\x\in\{-1,1\}^{n}}\max_{l\in\{1,2,\dots,m\}}\mathbf{a}_l^\mathsf{T}\x,
\end{equation*}
where $n=2N_t,$ $m=2K,$ and $\mathbf{a}_l^\mathsf{T}$ is the $l$-th row of $\mathbf{A}$.
In the following, we shall design algorithms based on the compact form (P), which appears to be easier to handle than  (P$_0$).
\vspace{-0.3cm}
\subsection{Proposed Negative $\ell_1$ Penalty Approach} One main difficulty of problem (P) lies in its discrete one-bit constraint. To deal with such difficulty, we resort to the penalty technique, which penalizes the constraint into the objective with some carefully selected penalty function.
Specifically,
the proposed approach relaxes the discrete one-bit constraint $\x\in\left\{-1,1\right\}^n$ into the continuous constraint $\x\in\left[-1, 1\right]^n$, and includes a \emph{negative $\ell_1$ penalty} into the objective as
\begin{equation*}
\text{(}\text{P}_\lambda\text{)} \qquad\min_{\x\in[-1,1]^{n}}\max_{l\in\{1,2,\dots,m\}}~\mathbf{a}_l^\mathsf{T}\x-\lambda\|\x\|_1,
\end{equation*}
where $\lambda\geq 0$ is the penalty parameter. Intuitively, the negative $\ell_1$ penalty term in (P$_\lambda$) encourages large magnitudes of $\left\{x_i\right\}.$

The following theorem establishes both the global and  local equivalence of the original problem (P) and the penalty model (P$_\lambda$). 
\begin{theorem}\label{equivalence}
If the penalty parameter $\lambda$ in (P$_{\lambda}$) satisfies $\lambda>\max_l\|\mathbf{a}_l\|_\infty$, then the following results hold:
\begin{enumerate}
\renewcommand{\labelenumi}{(\theenumi)}
\item[(i)] Any optimal solution of (P$_\lambda$) is also an optimal solution of (P), and vice versa.
\item[(ii)] Any local minimizer of (P$_\lambda$) is a feasible point of (P); on the other hand, any feasible point of (P) is a (strict) local minimizer of (P$_\lambda$).
\end{enumerate}
\end{theorem}
\begin{proof}
See Appendix \ref{AppendixC}.
\end{proof}
Theorem \ref{equivalence} shows that problems (P) and (P$_\lambda$) are equivalent in the sense that they share the same global and local solutions. {\color{black}Three} remarks on the equivalence result in Theorem \ref{equivalence} are in order.
 (i) This equivalence result is highly nontrivial from the optimization perspective, because it shows that problem (P$_\lambda$) is not only globally but also locally equivalent to (P). It is generally difficult to find such an exact penalty model, especially when the original problem is non-smooth (which is our case).
(ii) The equivalence result  serves as a (necessary) theoretical guarantee that  we can forget the original \emph{discrete} model (P), and focus on the \emph{continuous} model (P$_\lambda$) for the algorithmic design. This is important and beneficial due to the following reasons. First, it gives us more freedom to design algorithms, since continuous problems are generally easier to handle than discrete problems. Second  (and more importantly), solving the continuous problem (P$_\lambda$) is more likely to give us a high-quality solution because we are solving the problem in a larger searching space in which the homotopy (sometimes called warm-start) technique \cite{homotopy3,homotopy1,penaltyproof} can help to bypass bad local solutions (equivalently bad feasible solutions of (P)).  {\color{black}(iii)  Theorem \ref{equivalence} can be generalized and might be useful for other applications. Specifically, for any optimization problem with Lipschitz continuous objective and discrete constraint $\x\in\{-1,1\}^n$, we can obtain a corresponding negative $\ell_1$ penalty model, and the (global and local) equivalence between the two problems holds when $\lambda> L$, where $L$ is the Lipschitz constant of the objective function. }

{\blue We are now ready to give our proposed algorithmic framework for solving problem (P).
We employ the homotopy technique  together with the penalty method. }
Specifically, we solve problem (P$_\lambda$) with a small penalty parameter $\lambda$ at the beginning, then gradually increase the penalty parameter and trace the solution path of the corresponding penalty problems, until the penalty parameter is sufficiently large and a feasible point of (P) is obtained.
{\color{black} Empirically}, this homotopy technique can avoid bad local minima and exhibits better {\color{black} numerical performance} than using a fixed large $\lambda$.
{\blue For more information of homotopy optimization, please refer to \cite{homotopy3,homotopy1,penaltyproof}.}
We name the above procedure for solving problem (P) as the negative $\ell_1$ penalty (NL1P) approach and give the algorithmic framework as follows. {\color{black}Algorithm \ref{algonp} in the fourth line of Algorithm \ref{nl1p}, which is the algorithm designed for solving problem (P$_{\lambda}$), will be elaborated in Section \ref{AMPlambda}.}

 \vspace{0.3cm}
{\color{black}\begin{algorithm}[H]
	\caption{NL1P Approach for Solving Problem (P)}
	\small
	\begin{algorithmic}[1]
		\STATE \textbf{Input:} $\lambda^{(0)}, ~\delta>1, ~\x^{(0)}$.
		\STATE \textbf{Initialize:} $t=0$.
				\REPEAT
		\STATE     Apply Algorithm \ref{algonp} (see further ahead) to solve problem (P$_{\lambda}$) with parameter $\lambda=\lambda^{(t)}$ and initial point $\x^{(t)}$; let the solution be $\x^{(t+1)}$.

		\STATE Set $\tilde{\x}^{(t+1)}=\text{sgn}(\x^{(t+1)})$.
		\STATE 
		Set $\lambda^{(t+1)}=\delta\lambda^{(t)}$ and $t = t+1.$
		\UNTIL{$\x^{(t)}$ satisfies the one-bit constraint}.
		\STATE \textbf{Output:} $\tilde{\x}=\arg\min_{\x\in\{\tilde{\x}^{(1)},\tilde{\x}^{(2)},\dots,\tilde{\x}^{(t)}\}}\max_{l\in\{1,2,\dots,m\}}\mathbf{a}_l^T\x$.
	\end{algorithmic} \label{nl1p}
\end{algorithm}}
\vspace{0.2cm}
 \hspace{-0.35cm}{\color{black}Note that in the above algorithm, we quantize all intermediate points $\left\{\x^{(t)}\right\}$  to satisfy the one-bit constraint and choose the quantized point with the best function value as the final output. This simple technique can further improve the quality of the obtained solution.}



\subsection{Remarks on Proposed NL1P Approach}
In this subsection, we give some {\color{black}important remarks on the proposed NL1P approach. }

\subsubsection{Comparison with LP Relaxation Based Approaches}\label{compare1}
Most of the existing approaches {\color{black}for solving (P)} (e.g., MSM\cite{CIfirst,QCECI,CImodel}, OPSU\cite{PBB}, {\color{black}GMSM\cite{GMSM}}, P-BB\cite{PBB}) are based on the LP relaxation model, which corresponds to problem (P$_\lambda$) with $\lambda=0$. 
{\color{black} As discussed in Section \ref{relatedworks},} the LP relaxation based approaches suffer from a high computational complexity (e.g., P-BB) or their performance degrades in difficult cases (e.g., OPSU, {\color{black}GMSM}, and MSM). Essentially, this is because the LP relaxation model on which they are based is not equivalent to the original model and thus a second stage is still needed to determine the values of elements  that do not satisfy the one-bit constraint, which is independently very difficult.  In contrast, the proposed  NL1P approach seeks to solve the negative $\ell_1$ penalty model (P$_\lambda$), which is an exact reformulation of the original problem (P). This explains why our proposed approach outperforms the LP relaxation based approaches, as will be observed in the simulation.

\subsubsection{Comparison with the Work in \cite{GEMM}}\label{compare2}

It is interesting to note that, though with different motivations, problem (P) is in the same form as the problem considered in \cite{GEMM}, {\color{black} where one-bit precoding design for QAM modulation based on the SEP metric is studied.}
In contrast to our proposed approach that deals with the non-smooth objective function, the authors in \cite{GEMM} developed a penalty method based on a smooth approximation.
 Specifically, they applied the log-sum-exponential approximation to the non-smooth maximum function and considered the following  approximation problem of (P):
 \begin{equation}\label{approx}
 \min_{\mathbf{x}\in\{-1,1\}^n} \sigma \log\left(e^{\frac{\mathbf{a}_1^\mathsf{T}\mathbf{x}}{\sigma}}+e^{\frac{\mathbf{a}_2^\mathsf{T}\mathbf{x}}{\sigma}}+\dots+ e^{\frac{\mathbf{a}_m^\mathsf{T}\mathbf{x}}{\sigma}}\right).
 \end{equation}
For the above (smooth) approximation  problem, the following negative square penalty model was proposed by adding a penalty term $-\lambda\|\mathbf{x}\|_2^2$ to the objective and relaxing the discrete constraint into the continuous box constraint:
 \begin{equation}\label{nsp}
 \min_{\mathbf{x}\in[-1,1]^n} \sigma \log\left(e^{\frac{\mathbf{a}_1^\mathsf{T}\mathbf{x}}{\sigma}}+e^{\frac{\mathbf{a}_2^\mathsf{T}\mathbf{x}}{\sigma}}+\dots+ e^{\frac{\mathbf{a}_m^\mathsf{T}\mathbf{x}}{\sigma}}\right)-\lambda\|\mathbf{x}\|^2_2.
 \end{equation}
{\color{black}A problem associated with this smoothing-based approach is that} its performance is sensitive to the smoothing parameter $\sigma$. 
On the one hand, the smoothing parameter is expected to be as small as possible to obtain an accurate approximation. On the other hand, a small smoothing parameter will result in 
a large Lipschitz constant of the  gradient of the objective function in \eqref{approx} (which is proportional to $1/\sigma$). 
This will pose great challenge to the algorithmic design, since first-order algorithms generally converge slowly when  the Lipschitz constant of gradient is {\color{black}large.}
In addition, it is shown in  \cite[Theorem 2]{GEMM} that the negative square penalty model \eqref{nsp} is locally equivalent to   problem \eqref{approx} only when the penalty parameter $\lambda$ is larger than the Lipschitz constant of the gradient of the objective function in \eqref{approx}. This means that the penalty parameter $\lambda$ needs to be very large to guarantee the local equivalence if $\sigma$ is small, which leads to more numbers of  iterations if the homotopy technique (which increases $\lambda$ gradually) is employed. 
This reveals the dilemma of the choice of the smoothing parameter $\sigma$ in \eqref{approx}.

 In contrast, our proposed approach deals with the non-smooth objective \emph{directly}, and thus avoids  the above dilemma of choosing the smoothing parameter $\sigma$. Moreover, the required lower bound of the penalty parameter in our penalty model (P$_\lambda$) to guarantee the equivalence, i.e., $\max_l\|\mathbf{a}_l\|_\infty$,  is only related to the problem parameter. 
  Nevertheless, the resulting non-smooth penalty model (P$_\lambda$) seems more challenging to solve than the smooth penalty model in \cite{GEMM}. In the next section, we will propose an efficient algorithm for solving problem (P$_\lambda$) by taking care of its special structure.
\subsubsection{Necessity of the Negative $\ell_1$ Penalty}
One may ask why we do not add the negative square penalty to the objective as in \cite{GEMM}, whereby the resulting model is
\begin{equation}\label{squarepenalty}
\min_{\x\in[-1,1]^{n}}\max_{l\in\{1,2,\dots,m\}}~\mathbf{a}_l^\mathsf{T}\x-\lambda\|\x\|_2^2.
\end{equation}
Next we show that \eqref{squarepenalty}
is not a good penalty model for problem (P). Specifically, for any $\lambda>0$, local minimizers of \eqref{squarepenalty} are not necessarily feasible points of problem (P). We give an example as follows.
\begin{example}
Consider the following problem:
\begin{equation}\label{example}
\begin{aligned}
\min_{\x\in\{-1,1\}^2}~&\max_{l\in\{1,2,3,4\}}~\mathbf{a}_l^\mathsf{T}\x\\
\end{aligned}
\end{equation}
where
$$\A=\left(\begin{matrix} \mathbf{a}_1^\mathsf{T}\\\mathbf{a}_2^\mathsf{T}\\\mathbf{a}_3^\mathsf{T}\\\mathbf{a}_4^\mathsf{T}\end{matrix}\right)=\left(\begin{matrix} 1&-1\\-1&1\\1&1\\-1&-1\end{matrix}\right).$$
The corresponding negative square penalty problem is
\begin{equation}\label{l2penalty}
\begin{aligned}
\min_{\x\in[-1,1]^2}~&\max_{l\in\{1,2,3,4\}}~\mathbf{a}_l^\mathsf{T}\x-\lambda\|\x\|_2^2\\
\end{aligned}
\end{equation}

We claim that for any $\lambda>0$, $\mathbf{0}$ is a local minimizer of (\ref{l2penalty}) (but not a feasible point of (\ref{example})).
Note that (\ref{l2penalty}) is equivalent to
\begin{equation*}
\begin{aligned}
\min_{\x\in[-1,1]^2}~&\|\x\|_1-\lambda\|\x\|_2^2\\
\end{aligned}
\end{equation*}
Given any $\lambda>0$, then for all $\x\in B\left(\mathbf{0},\frac{1}{\lambda}\right)$, it holds that
$$\|\x\|_1-\lambda\|\x\|_2^2\geq \|\x\|_2-\lambda\|\x\|_2^2=\|\x\|_2(1-\lambda\|\x\|_2)\geq 0,$$ which implies that $\mathbf{0}$ is a local minimizer of (\ref{l2penalty}).

\end{example}

The above example shows the failure of the negative square penalty and reveals the difficulty of finding an exact penalty for a non-smooth problem like (P). The main reason for the failure of the negative square penalty lies in that a smooth penalty is utilized in the problem where the objective is non-smooth.  This also explains why we choose the non-smooth negative $\ell_1$ penalty for problem (P).
 In fact, the negative $\ell_1$ penalty we adopt in this paper is the simplest penalty we can find that guarantees the local equivalence.

{\color{black}\subsubsection{Applications to More Practical Communication Scenarios}
While initially designed for solving the one-bit precoding problem in a flat-fading channel with PSK signaling, the proposed NL1P approach can be directly applied or easily generalized to various other communication scenarios.
\begin{itemize}
\item \emph{QAM constellation.} As has been discussed before,  the one-bit precoding problem for QAM constellation has been formulated in the same form as (P) in \cite[Eq. (16)]{GEMM}. The proposed NL1P approach can be directly applied to solve the model in \cite{GEMM}.
\item \emph{OFDM system.} The CI-based model in \cite{CIfirst}  has been generalized to  the OFDM system  \cite[Eq. (20)]{OFDM}. Following the same idea, the symbol scaling model (P$_0$) can also be generalized to the OFDM system and our proposed NL1P approach is still applicable. 
\item \emph{Quantized constant envelope (QCE) precoding.} For the more general QCE precoding case, the real and imaginary parts of each transmit signal couples in the QCE constraint, and thus the real-space QCE constraint is  two-dimensional  and cannot be further decomposed as in the one-bit case; see \cite[Eq. (12)]{QCECI}. In this case, the negative $\ell_1$ penalty does not work any more. However, we could make a slight modification to change the negative $\ell_1$ norm into the sum of the negative $\ell_2$ norm,  where each $\ell_2$ norm is introduced to penalize each two-dimensional QCE constraint. Please see more details on this generalization in \cite{ICASSPQCE}.
\end{itemize}
}



\vspace{-0.12cm}
\section{An Efficient Algorithm for solving problem (P$_\lambda$)}\label{algorithm}
In this section, we propose an efficient algorithm for solving the non-smooth non-convex subproblem (P$_\lambda$) in the NL1P approach. More specifically, we first transform problem (P$_\lambda$) into an equivalent min-max problem ($\widehat{\text{P}}_\lambda$) in Section \ref{minmaxreform}. Then we propose an efficient alternating proximal/projection gradient descent ascent (APGDA) algorithm for solving a class of non-smooth min-max problems (which includes our problem ($\widehat{\text{P}}_\lambda$) as a special case) and give the convergence analysis in Section \ref{AM} and Section \ref{AMconverge}, respectively. In Section \ref{AMPlambda}, we apply the proposed APGDA algorithm to solve problem ($\widehat{\text{P}}_\lambda$) and give some discussions.
\vspace{-0.15cm}
\subsection{Min-Max Reformulation of (P$_{\lambda}$)}\label{minmaxreform}
In this subsection, we reformulate problem (P$_\lambda$) into an equivalent min-max problem.
Recall that the objective in (P$_\lambda$) is the maximum of a finite collection of functions. By introducing an auxiliary variable
\begin{equation}\label{delta}
\y\in\Delta\triangleq\left\{\y\in\R^m\mid\mathbf{1}^\mathsf{T}\y=1, \y\geq\mathbf{0}\right\},
\end{equation}
 (P$_\lambda$) can be equivalently transformed into the following min-max problem:
\begin{equation*}
\text{(}\widehat{\text{P}}_\lambda\text{)}\quad\min_{\x\in[-1,1]^n}~\max_{\y\in\Delta}~\y^\mathsf{T}\A\x-\lambda\|\x\|_1.
\end{equation*}
 The two problems (P$_\lambda$) and ($\widehat{\text{P}}_\lambda$) are equivalent in the sense that an optimal solution (stationary point) of one problem can be easily constructed given an optimal solution  (stationary point) of the other problem \cite{HiBSA}.

Below we shall focus on designing an efficient algorithm for solving the reformulated min-max problem ($\widehat{\text{P}}_\lambda$). In the next subsection, we shall develop an algorithm for solving a class of non-smooth nonconvex-concave min-max problems, which includes ($\widehat{\text{P}}_\lambda$) as a special case.
\vspace{-0.15cm}
\subsection{Proposed Algorithm}\label{AM}
Min-max problems have drawn considerable interest (especially in machine learning and signal processing communities) in recent years. Various algorithms have been proposed for different types of min-max problems\cite{xu2020unified,MMPD,zhang2020singleloop,pmlr-v119-lin20a,HiBSA,dai2020majorized}. However, previous works mainly consider the smooth case\cite{MMPD,xu2020unified,zhang2020singleloop,pmlr-v119-lin20a}. Few works that focus on non-smooth min-max problems all require the non-smooth term to be convex\cite{HiBSA}\cite{dai2020majorized}.  To the best of our knowledge, there is no existing work that covers {\color{black}our problem of interest} ($\widehat{\text{P}}_\lambda$) (where the negative $\ell_1$ penalty term in the objective is both non-smooth and non-convex), and thus no existing algorithm can be directly applied to solve problem ($\widehat{\text{P}}_\lambda$).

In this subsection, we consider a class of non-smooth nonconvex-concave min-max problems
\begin{equation}\label{ourprob}
\min_{\x\in \mathcal{X}}\max_{\y\in \mathcal{Y}}~F(\x,\y)\triangleq f(\x,\y)-g(\x),
\end{equation}
where $f(\x,\y)$ is a smooth function that is non-convex with respect to $\x$ and concave with respect to $\y$,  $g(\x)$ is a non-smooth, proper closed convex function,  and $\mathcal{X}$ and $\mathcal{Y}$ are compact convex sets in $\R^n$ and $\R^m$, respectively. Problem \eqref{ourprob} includes problem ($\widehat{\text{P}}_\lambda$) as a special case. To be specific,  $f(\x,\y)$ and $g(\x)$ correspond to the linear term $\y^\mathsf{T}\A\x$ and  the $\ell_1$ norm $\|\x\|_1$, respectively; $\mathcal{X}$ and $\mathcal{Y}$ correspond to $[-1,1]^n$ and the simplex set $\Delta$ in \eqref{delta}, respectively.

Our proposed algorithm for solving problem \eqref{ourprob} can be regarded as an extension of the  algorithms proposed in \cite{HiBSA} and \cite{xu2020unified} from the smooth case to the non-smooth case, which is independently interesting. In \cite{HiBSA} and \cite{xu2020unified}, the authors proposed  unified frameworks for solving a few different classes of min-max problems including the smooth nonconvex-concave ones, which is a special case of \eqref{ourprob} with $g(\x)=0$.

Similar to \cite{HiBSA} and \cite{xu2020unified}, a perturbed function of the original objective:
\begin{equation*}
\tilde{F}(\x,\y)=F(\x,\y)-\frac{c_k}{2}\|\y\|_2^2=f(\x,\y)-g(\x)-\frac{c_k}{2}\|\y\|_2^2
\end{equation*}
is considered, where the perturbed term is introduced to make $\tilde{F}(\x,\y)$ strongly concave in $\y$. It is shown in \cite{HiBSA} and \cite{xu2020unified} that the perturbed term is important for the convergence of the corresponding algorithms.

At each iteration, the proposed algorithm updates $\x$ and $\y$ alternately as follows: 
\begin{subequations}\label{12}
\begin{align}
\x_{k+1}&\in\arg \min _{\x \in \mathcal{X}}~f(\x_k,\y_k)+\left<\nabla_{\x} {f}(\x_{k}, \y_{k}), \x-\x_{k}\right>-g(\x)-\frac{c_k}{2}\|\y_k\|_2^2+\frac{\tau_{k}}{2}\left\|\x-\x_{k}\right\|_2^{2},\label{updatexx}\\
\y_{k+1}&=\arg\max_{\y\in\mathcal{Y}}~\tilde{F}(\x_{k+1},\y_k)+\left<\nabla_{\y}\tilde{F}(\x_{k+1},\y_k),\y-\y_k\right>-\frac{1}{2\rho_k}\|\y-\y_k\|_2^2\notag\\
&=\text{Proj}_{\mathcal{Y}}\left(\y_k+\rho_k\nabla_{\y} f(\x_{k+1},\y_k)-\rho_kc_k\y_k\right),\label{updateyy}
\end{align}
\end{subequations}
where $\rho_k > 0$ and $\tau_k> 0$ are the properly selected regularization parameters. Generally, the solution to the $\x$-subproblem might not be unique, and in this case we only need to choose {\color{black} an arbitrary} one  from the solution set. We summarize the above algorithm as Algorithm \ref{algorithm1} below. 
\begin{algorithm}
\caption{APGDA Algorithm for Solving Problem \eqref{ourprob}}
\begin{algorithmic}[1]\label{algorithm1}
\small
	\STATE \textbf{Input:} $\x_0,~\y_0,~\{\tau_k\},~\{\rho_k\},~\{c_k\}$.
		\STATE \textbf{Initialize:} $k=0$.
\REPEAT
\STATE Alternately update $\x_{k+1}$ and $\y_{k+1}$ as in \eqref{updatexx} and \eqref{updateyy}.
\STATE Set $k=k+1$.
\UNTIL some stopping criterion is satisfied.
\STATE {\color{black}\textbf{Output:} $\x_{k}$}.
\end{algorithmic}
\end{algorithm}

Some remarks on the proposed Algorithm \ref{algorithm1} and parameters in it are as follows. 
{\color{black}The $\x$-subproblem \eqref{updatexx} can be seen as minimizing a local approximation of $\tilde{F}(\cdot,\y_k)$ around $\x_k$, or more precisely,  a quadratic approximation of the smooth term in $\tilde{F}(\cdot,\y_k)$ around $\x_k$, i.e., $f(\x_k,\y_k)+\left<\nabla_{\x} {f}(\x_{k}, \y_{k}), \x-\x_{k}\right>+\frac{\tau_{k}}{2}\left\|\x-\x_{k}\right\|_2^{2}-\frac{c_k}{2}\|\y_k\|_2^2$, plus the non-smooth term in $\tilde{F}(\cdot,\y_k)$, i.e., $-g(\x)$.  This idea is actually the same as that of the proximal gradient method \cite{proximal} and \eqref{updatexx} can be seen as a proximal gradient step. (The only difference is that $-g(\x)$ is concave here while the traditional proximal gradient method deals with convex functions.) }
Similarly, $\y$ is updated via a classical gradient projection step of the perturbed function. Since Algorithm \ref{algorithm1} updates variable $\x$ by performing a proximal gradient step and variable $\y$ by performing a projection gradient step in an alternating fashion, we name it as the alternating proximal/projection gradient descent ascent (APGDA) algorithm. 
 The parameters {\color{black}$\{\tau_k\}$ and $\{\rho_k\}$ in \eqref{12} are the stepsizes of the proximal/projection gradient steps}, 
 and $\{c_k\}$ controls the  accuracy and strong concavity of the perturbed function. Properly selecting those parameters plays a vital role in guaranteeing the convergence and good performance of the proposed algorithm.

The efficiency of the proposed APGDA algorithm depends on the efficiency of solving the subproblems in \eqref{12}. The $\x$-subproblem \eqref{updatexx}  is a non-smooth non-convex problem, which generally does not admit a closed-form solution. However, for many cases of our interest, the $\x$-subproblem \eqref{updatexx} either has a closed-form solution or can be efficiently solved to high accuracy. For instance, if $\mathcal{X}$ is a  Cartesian product of $n$ simple compact convex sets, i.e., $\mathcal{X}=\prod_{i=1}^n\mathcal{X}_i$, and $g(\x)$ is simple and separable in $\x$, i.e., $g(\x)=\sum_{i=1}^n g_i(x_i)$, then the exact solution can be obtained by solving $n$ simple one-dimensional problems. Fortunately, the interested problem ($\widehat{\text{P}}_\lambda$) is such a case and we shall give a detailed discussion on this in Section \ref{AMPlambda} later on. The $\y$-subproblem \eqref{updateyy} is a projection problem onto set $\mathcal{Y}$ and can be efficiently solved for many cases of $\mathcal{Y}$ such as the simplex set $\Delta$ in \eqref{delta}.

\vspace{-0.2cm}
 \subsection{Convergence Analysis}\label{AMconverge}
In this subsection, we establish the global convergence of the proposed APGDA algorithm. Before doing this, we give the following definition of the stationary point, which is a generalization of \cite[Definition 3.1]{zhang2020singleloop} from the smooth case to the non-smooth case.
\newtheorem{definition}{Definition}
\begin{definition}\label{stationarydef}
A pair $(\hat{\x},\hat{\y})$ is called a stationary point of problem (\ref{ourprob}) if
\begin{equation*}\label{conditionw}
\left\{
\begin{aligned}
\mathbf{0}\in&\,\nabla_{\x} f(\hat{\x},\hat{\y})-\partial g(\hat{\x})+\partial \mathbb{I}_\mathcal{X}(\hat{\x});\\
\mathbf{0}\in&-\nabla_{\y} f(\hat{\x},\hat{\y})+\partial \mathbb{I}_
\mathcal{Y}(\hat{\y}),
\end{aligned}
\right.
\end{equation*}
 where $\mathbb{I}_\mathcal{X}(\cdot)$ and $\mathbb{I}_\mathcal{Y}(\cdot)$ are the indicator functions of $\mathcal{X}$ and $\mathcal{Y}$, respectively. \end{definition}

To establish the convergence, we need to impose the following assumptions on $f$ and $g$ in problem \eqref{ourprob}.

\begin{assumption}\label{assumption1}
The function $f(\x,\y)$ is continuously differentiable and there exist constants $L_x, L_{21}, L_y$, and $L_{12}$ such that for $\x,\x_1,\x_2\in \mathcal{X}$ and $\y, \y_1, \y_2\in \mathcal{Y}$, we have
\begin{equation*}
\begin{aligned}
\|\nabla_{\x} f(\x_1,\y)-\nabla_{\x} f(\x_2,\y)\|_2&\leq L_{x}\|\x_1-\x_2\|_2,\\
\|\nabla_{\x} f(\x,\y_1)-\nabla_{\x} f(\x,\y_2)\|_2&\leq L_{21}\|\y_1-\y_2\|_2,\\
\|\nabla_{\y} f(\x,\y_1)-\nabla_{\y} f(\x,\y_2)\|_2&\leq L_y\|\y_1-\y_2\|_2,\\
\|\nabla_{\y} f(\x_1,\y)-\nabla_{\y} f(\x_2,\y)\|_2&\leq L_{12}\|\x_1-\x_2\|_2.\\
\end{aligned}
\end{equation*}
\end{assumption}
\begin{assumption}\label{assumptionong}
The function $g(\x)$ is Lipschitz continuous, i.e., {\color{black}there exists $G>0$ such that}
$$|g(\x_1)-g(\x_2)|\leq G\|\x_1-\x_2\|_2, \quad \forall~ {\color{black}\x_1,\x_2\in\mathbb{R}^n}.$$
\end{assumption}

With the above definition and assumptions, we are ready to present the convergence result of the proposed APGDA algorithm. 
\vspace{-0.3cm}
\begin{theorem}\label{convergethe}
Suppose that Assumptions 1 and 2 hold. Let $\{(\x_k,\y_k)\}$ be the sequence generated by Algorithm \ref{algorithm1} with $\rho_k=\rho$.
If  $0<\rho\leq \frac{2}{L_y+2\beta_1}$, $c_k=\frac{\beta_1}{(k+1)^\gamma}$ with $0<\gamma\leq 0.5$, $\beta_1>0$, and  $\tau_k=\frac{16\beta_2 L_{12}^2}{\rho c_k^2}+\beta_3$ with $\beta_2>1$, $\beta_3\geq\rho L_{12}^2+L_x$, then any limit point of $\{(\x_k,\y_k)\}$ is a stationary point of problem \eqref{ourprob}.
\end{theorem}
\begin{proof}
See Appendix \ref{AppendixD}.
\end{proof}
\vspace{-0.45cm}
\subsection{APGDA Algorithm for Solving ($\widehat{\text{P}}_\lambda$)}\label{AMPlambda}
In this subsection, we specialize the proposed APGDA algorithm to problem ($\widehat{\text{P}}_\lambda$) and carefully investigate its behaviors on this special problem, including implementation details (see Algorithm \ref{algonp}) and convergence results. We also propose a low-complexity implementation of  Algorithm \ref{algonp} to further reduce the computational cost. 
\subsubsection{Implementation Details}\label{compare3}
Specializing Algorithm \ref{algorithm1} to problem ($\widehat{\text{P}}_\lambda$),  the subproblems of $\x$ and $\y$ become
\begin{equation}\label{updatex}
\x_{k+1}\in\arg\min_{\x\in[-1,1]^n}\y_k^\mathsf{T}\A\x-\lambda\|\x\|_1+\frac{\tau_k}{2}\|\x-\x_k\|_2^2
\end{equation}
 and
 \begin{equation}\label{updatey}\y_{k+1}=\text{Proj}_{\Delta}\left(\y_k+{\rho_k}\A\x_{k+1}-{\rho_k}c_k\y_k\right).\end{equation}
The $\x$-subproblem \eqref{updatex} is separable and has a closed-form solution. More specifically, by denoting
$\A=[\A_1,\A_2,\dots, \A_{n}]$ {\color{black} with $\A_i$ representing the $i$-th column of $\A$}, the subproblem \eqref{updatex} decouples into $n$  of problems in the following form:
\begin{equation}\label{distributedform}
\begin{aligned}
\x_{k+1}(i)\in\arg\min_{-1\leq x \leq 1}(\A_i^\mathsf{T}{\y}_k)x-\lambda |x|+\frac{\tau_k}{2}(x-\x_k(i))^2,~ i=1,2,\ldots,n,
\end{aligned}
\end{equation}
which admits a closed-form solution as
\begin{equation}\label{solutionx}
\x_{k+1}(i)=\text{sgn}(a_k^i)\min\left\{|a_k^i|+\frac{\lambda}{\tau_k},1\right\},~i=1,2,\dots,n,
\end{equation}
where $a_k^i=\x_{k}(i)-\frac{\A_i^\mathsf{T}\y_k}{\tau_k}$. A detailed derivation of \eqref{solutionx} is given in Appendix \ref{AppendixE}. Note that when $a_k^i=0,$ the solution of \eqref{distributedform} is not unique and we only need to choose one from the solution set. Here we choose $\x_{k+1}(i)=\min\left\{\frac{\lambda}{\tau_k},1\right\}$, and thus the solution of \eqref{distributedform}  can be expressed in a unified way as \eqref{solutionx}.
The solution of $\y$-subproblem \eqref{updatey} involves only matrix-vector multiplications and a projection onto the  simplex, which has a very fast implementation \cite{projection2}\cite{projection}.
We summarize the specialization of  the APGDA algorithm for solving problem ($\widehat{\text{P}}_\lambda$) as Algorithm
\ref{algonp}.
\begin{algorithm}
\small
\caption{APGDA Algorithm for Solving Problem ($\widehat{\text{P}}_\lambda$)}
\begin{algorithmic}[1]\label{algonp}
	\STATE \textbf{Input:} $\x_0,~\y_0,~\{\tau_k\},~\{\rho_k\},~\{c_k\}$.
		\STATE \textbf{Initialize:} $k=0$.
\REPEAT
\STATE Alternately update $\x_{k+1}$ and $\y_{k+1}$ as in \eqref{solutionx} and \eqref{updatey}.
\STATE Set $k=k+1$.
\UNTIL some stopping criterion is satisfied.
\STATE {\color{black}\textbf{Output:} $\x_{k}$}.
\end{algorithmic}
\end{algorithm}

 In total, the dominant complexity of Algorithm \ref{algonp} at each iteration lies in  calculating $\mathbf{A}^\mathsf{T}\mathbf{y}$ and $\mathbf{A}\mathbf{x}$, which requires $2mn$ real-number multiplications, and computing one projection onto the simplex of dimension $m$, whose computational complexity is $\mathcal{O}(m\log m)${\color{black}\cite{projection2}\cite{projection}.}
 Recalling that $n=2N_t$ and $m=2K$.   That is to say, for a system with $N_t$ transmit antennas and $K$ users, the per-iteration complexity of the proposed  APGDA algorithm is $\mathcal{O}(N_tK+K\log K)$. 

\subsubsection{Convergence Behavior}
According to Theorem \ref{convergethe}, the APGDA algorithm (with properly chosen parameters) is guaranteed to find a stationary point of problem ($\widehat{\text{P}}_\lambda$), whose corresponding $\x$-part is also a stationary point of problem (P$_\lambda$) due to the equivalence between problems (P$_\lambda$) and ($\widehat{\text{P}}_\lambda$). The remaining question is whether the obtained stationary point satisfies the one-bit constraint. This is a crucial question, since if the obtained solution does not satisfy the one-bit constraint,  we have to further consider how to determine the values of those infeasible elements to obtain a high-quality one-bit solution. In fact, this is the problem faced in the second stage of the LP relaxation based approaches, which is independently  very difficult. Fortunately, by carefully exploiting the  special structures of problem (P$_\lambda$) and  Algorithm \ref{algonp}, we can give an affirmative answer to the above question, i.e., the obtained solution already satisfies the one-bit constraint. 
 To show this, we first characterize the stationary points of (P$_\lambda$).
\begin{lemma}\label{criticalthe}
If $\lambda>\max_l\|\mathbf{a}_l\|_\infty,$ all stationary points $\hat{\x}$ of (P$_\lambda$) must satisfy
\begin{equation*}
\hat{x}_i\in\left\{-1,1,0\right\},~i=1,2,\dots,n.
\end{equation*}
\end{lemma}
\begin{proof}
See Appendix \ref{AppendixF}.
\end{proof}
Theorem \ref{convergethe} and Lemma \ref{criticalthe} suggest that every limit point $\hat{\x}$ of the sequence generated by Algorithm \ref{algonp} must have all of its elements being either $\pm 1$ or $0$.
Obviously zero elements here do not satisfy the one-bit constraint in problem (P) and thus are undesirable. Fortunately, the following Theorem \ref{c1} shows that zero elements will not happen in Algorithm \ref{algonp}. Note that for problem ($\widehat{\text{P}}_\lambda$), $L_x = L_y = 0, L_{12} = L_{21}= \|\A\|_2$. The following theorem is a combination of results in Theorem \ref{convergethe}, Lemma \ref{criticalthe}, and the closed-form solution \eqref{solutionx}. 
\begin{theorem}\label{c1}
Let $\{(\x_k,\y_k)\}$ be the sequence generated by Algorithm \ref{algonp} with $\rho_k=\rho, c_k=\frac{\beta_1}{ (k+1)^{\gamma}}$, and $\tau_k=\frac{16\beta_2\|\A\|_2^2}{\rho c_k^2}+\beta_3$, where $0<\rho\leq\frac{1}{\beta_1}$, $0<\gamma\leq 0.5$, $\beta_1>0$, $\beta_2>1$, and $\beta_3\geq\rho\|\A\|_2^2$. Then if $\lambda>\max_l\|\mathbf{a}_l\|_\infty,$ every limit point $\hat{\x}$ of $\{\x_k\}$ must satisfy the one-bit constraint.
\end{theorem}
\begin{proof}
See Appendix \ref{AppendixG}.
\end{proof}
In summary, when the penalty parameter $\lambda$ in problem (P$_\lambda$) is  sufficiently large, every limit point $\hat{\x}$ of the sequence generated by Algorithm \ref{algonp} (with properly selected parameters) is not only a stationary point of (P$_\lambda$) but also a feasible point of (P) and thus a local minimizer (according to Theorem \ref{equivalence}) of problem (P$_\lambda$). This nice property is a result of the combination of nice properties of problem (P$_\lambda$) and Algorithm \ref{algonp}. {\color{black} In comparison, the GEMM algorithm proposed in \cite{GEMM} is not guaranteed to terminate at a feasible point of (P), which shows the theoretical superiority of our proposed algorithm.}
\subsubsection{A Low-Complexity Implementation of Algorithm \ref{algonp}}
To further reduce the computational cost, in this part we propose a  low-complexity implementation of Algorithm \ref{algonp}. To be specific, we consider performing Algorithm \ref{algonp} in a more aggressive manner by keeping the values of variables fixed in later iterations once they satisfy the one-bit constraint. For clarity, we summarize the above procedure in Algorithm \ref{anl1pinner}.
\begin{algorithm}
\caption{\small A Low-Complexity Implementation of Algorithm \ref{algonp}}
\label{anl1pinner}
\begin{algorithmic}[1]
\small
	\STATE \textbf{Input:} $\x_0,~\y_0,~\{\tau_k\},~\{\rho_k\},~\{c_k\}$.
		\STATE \textbf{Initialize:} $k=0$,  $\mathcal{S}=\left\{i\in\{1,2,\dots,n\}\mid |\x_0(i)|<1\right\}$.
\REPEAT

\STATE  Update $\x_{k+1}$ as
\begin{equation*}
\x_{k+1}(i)=\left\{
\begin{aligned}
\text{sgn}(a_k^i)\min\left\{|a_k^i|+\frac{\lambda}{\tau_k},1\right\},~&\text{if }i \in \mathcal{S};\\
\x_{k}(i),\hspace{2cm}~&\text{otherwise},
\end{aligned}
\right.
\end{equation*} where
$a_k^i=\x_{k}(i)-\frac{\A_i^\mathsf{T}\y_k}{\tau_k}$.
\STATE  Update $\mathcal{S}$ as $\mathcal{S}=\left\{i\in\left\{1,2,\dots,n\right\}\mid |\x_{k+1}(i)|<1\right\}$.
\STATE  Update $\y_{k+1}$ as in \eqref{updatey}.
\STATE Set $k=k+1$.
\UNTIL some stopping criterion is satisfied.
\STATE {\color{black}\textbf{Output:} $\x_{k}$.}
\end{algorithmic}
\end{algorithm}

 If Algorithm \ref{anl1pinner} is employed to solve the subproblem (P$_\lambda$) in Algorithm 1,
 then the number of elements in $\x$ that need to be updated will gradually decrease as the algorithm proceeds.
 Therefore, replacing Algorithm \ref{algonp} with Algorithm \ref{anl1pinner} to solve the subproblem (P$_\lambda$) can accelerate the convergence of the NL1P approach. We name the corresponding algorithm as the accelerated negative $\ell_1$ penalty (ANL1P) approach. It will be shown in the simulation that  ANL1P can achieve almost the same performance as NL1P with less CPU time.

\section{Simulation Results}\label{simulation}
In this section, we present simulation results to demonstrate the performance of our proposed algorithms.
\vspace{-0.4cm}
\subsection{Simulation Setup and Choice of Parameters}
 We consider multiuser massive MIMO systems where the BS is equipped with hundreds of antennas.
We assume standard Rayleigh fading channel, i.e., the channel matrix $\H$ is composed of independent and identically distributed standard complex Gaussian random variables. 
We set the length of the transmission block to be $T=10$ and define the transmit SNR to be $\frac{1}{\sigma^2}$, where the unit transmit power is assumed. 
All the results are obtained with Monte Carlo simulations of 1000 independent channel realizations. 
{\color{black}Throughout this section, we use the triple $(K,N_t,M)$ to describe the considered system, where $K$ and $N_t$ denote the numbers of users and transmit antennas in the system, respectively, and $M$ refers to the constellation level for PSK modulation.}

We compare the proposed NL1P and ANL1P approaches with existing state-of-the-art linear and nonlinear one-bit precoding approaches listed in Table \ref{Compared algorithms}. {\color{black}
The GMSM algorithm in \cite{GMSM} exhibits a similar rationale and performance to the OPSU algorithm in \cite{PBB} and thus is omitted in the simulations for conciseness of presentation. In addition to the algorithms in Table \ref{Compared algorithms}, we also modify and apply the GEMM algorithm proposed in \cite{GEMM}, which is originally designed for QAM modulation, to solve our considered model, and we term it as ``CI 1-Bit GEMM''. }Moreover,
 We also include the unquantized ZF precoder, termed as ``Inf-Bit ZF'',  as the performance upperbound of the one-bit precoding approaches.

The parameters used in our algorithms are as follows. In Algorithm \ref{nl1p}, the initial point is {\color{black}set to}  $\x^{(0)}=\mathbf{0}$; the penalty parameter is initialized as $\lambda^{(0)}=\frac{0.001M}{8}$ and increased by a factor of $\delta=5$ at each iteration. In Algorithm \ref{algonp} (Algorithm \ref{anl1pinner}), we set the initial point of $\y$ as $\y_0=\frac{1}{2K}\mathbf{1}$ and the other parameters  as $\rho_k=\rho=\frac{0.2}{\|\A\|_2},~c_k=\frac{0.01}{\rho (k{\color{black}+1})^{0.05}},$ and $ \tau_k=1.2\text{mean}\left(|\A|\right)({\color{black}k+1})^{0.1}$. We terminate Algorithm \ref{algonp}  (Algorithm \ref{anl1pinner}) for solving the subproblem in Algorithm \ref{nl1p} when its iteration number is more than $500$ or when the distance of its successive iterates is less than $10^{-3}$. 
\vspace{-0.2cm}
\subsection{BER Performance}

  \begin{figure}[!t]
\centering
\includegraphics[scale=0.45]{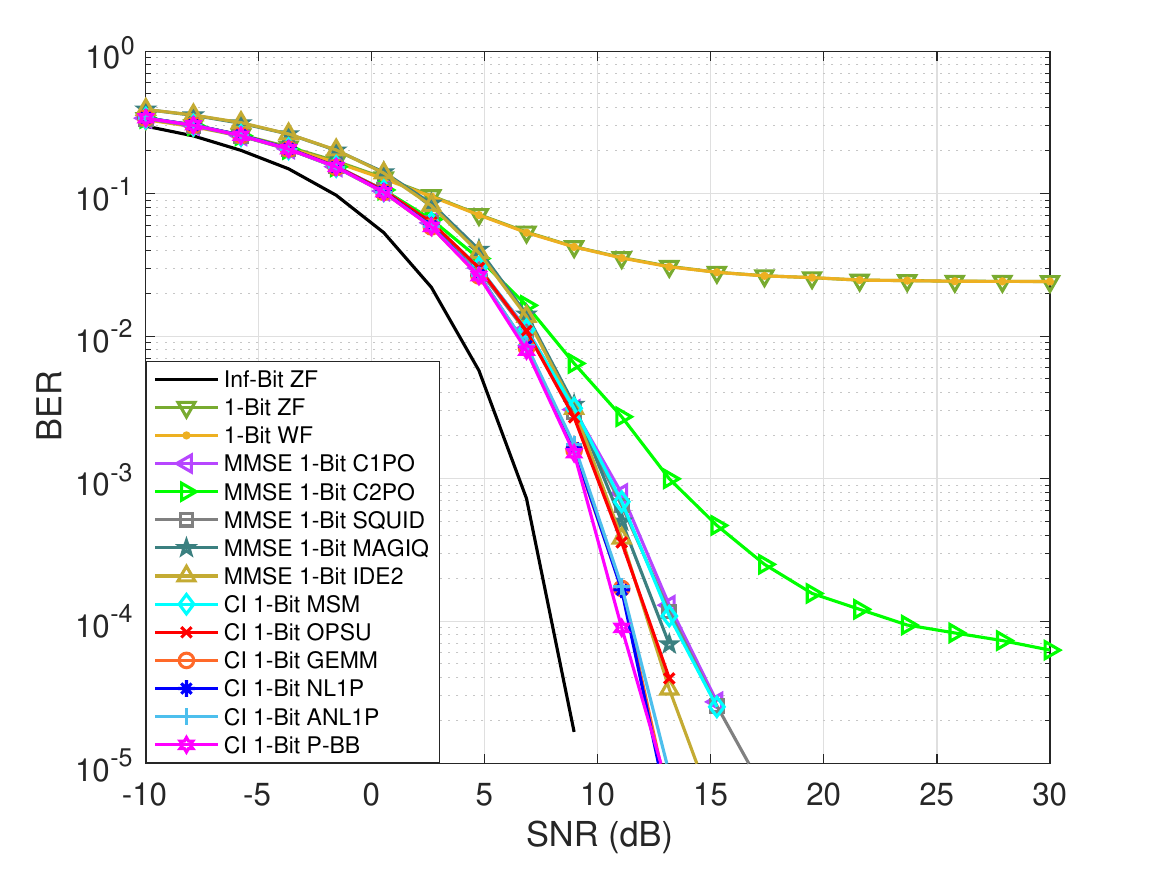}
\vspace{-10pt}
\caption{BER performance versus SNR, {\color{black} where $(K,N_t,M)=(16,128,8)$}.}
\label{16_128_8}
\vspace{-0.3cm}
\end{figure}
We first present the bit error rate (BER) results for different multiuser massive MIMO systems.
 In Fig. \ref{16_128_8}, a $16\times 128$ system with $8$-PSK modulation is considered. It can be observed that linear precoding suffers a BER floor in the high SNR regime due to the coarse one-bit quantization,  while all of the nonlinear approaches exhibit significantly better BER performance. Of the presented nonlinear precoding schemes, the CI-based methods generally perform better than the  MMSE-based methods, among which the P-BB algorithm achieves the best BER performance. However, since a branch-and-bound process is included, the P-BB algorithm is computationally inefficient, 
 as will be demonstrated in Section \ref{cputime}.  It can be observed from Fig. \ref{16_128_8} that all the CI-based approaches achieve comparable BER performance in this system, with the two proposed algorithms {\color{black} and the GEMM algorithm} showing  slightly better  performance than the state-of-the-art OPSU precoder.

\begin{figure*}[!t]
\subfigure[${\color{black}(K,N_t,M)=(32,128,8)}$.]{
\centering
\includegraphics[scale=0.402]{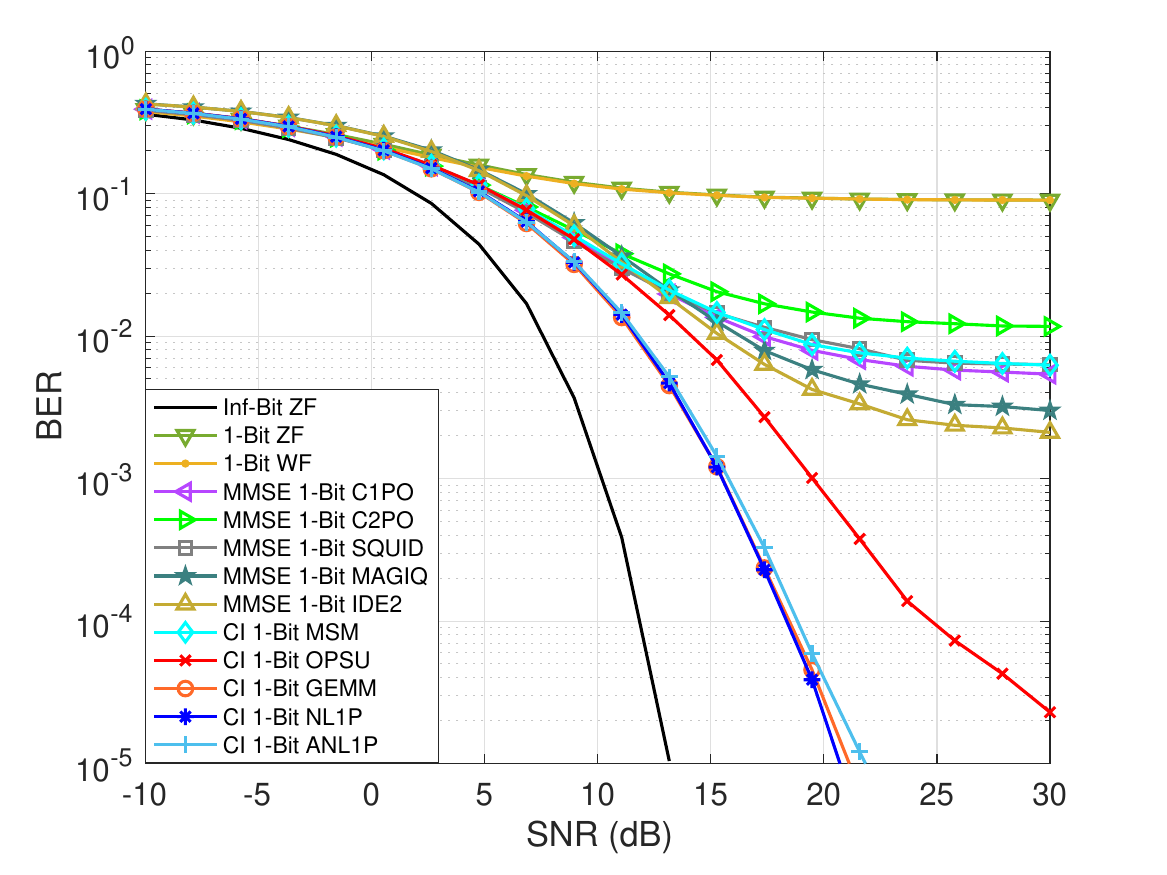}
\vspace{-10pt}
\label{32_128_8}}
\subfigure[${\color{black}(K,N_t,M)=(16,128,16)}$]{\includegraphics[scale=0.402]{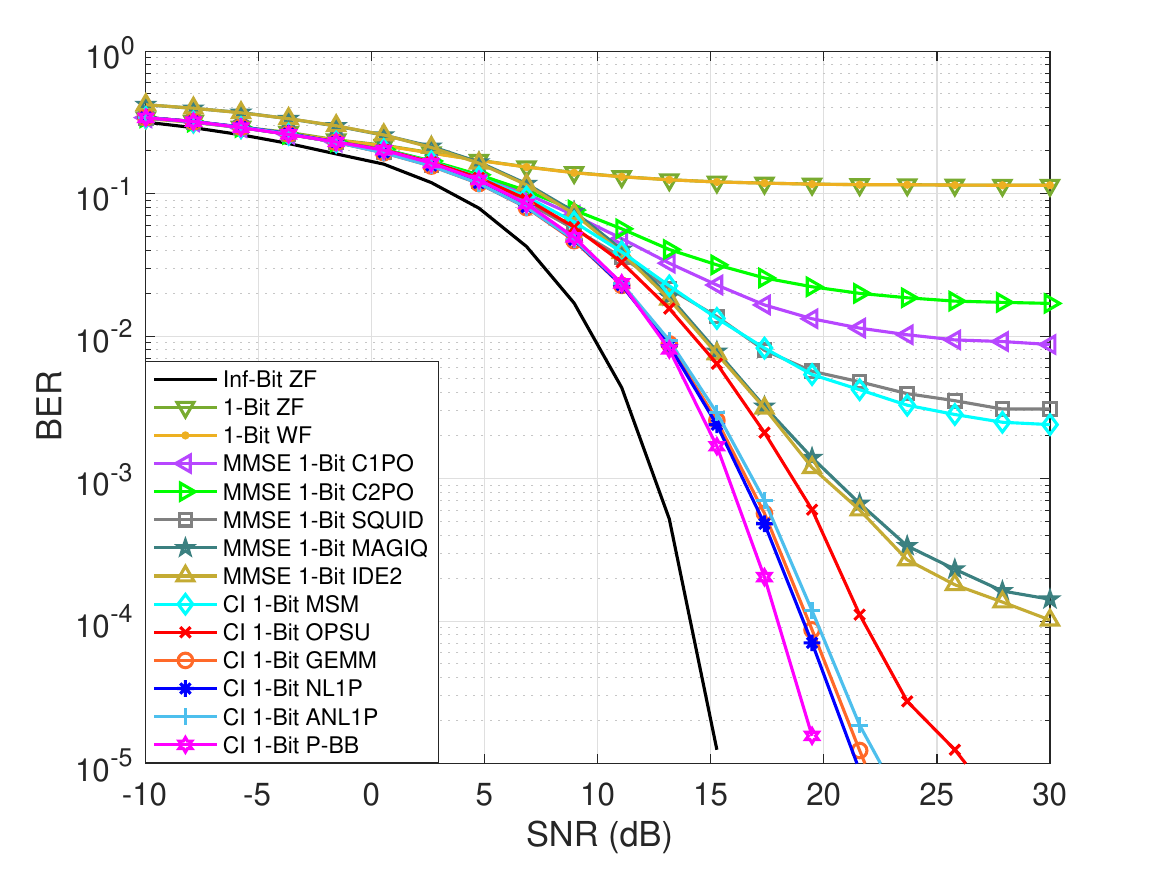}
\label{16_128_16}}
\vspace{-0.15cm}
\caption{{\color{black}BER performance versus SNR for high antenna-user ratio and high-level modulation}.}
\label{time}
\end{figure*}

\begin{figure}[!t]
\centering
\includegraphics[scale=0.45]{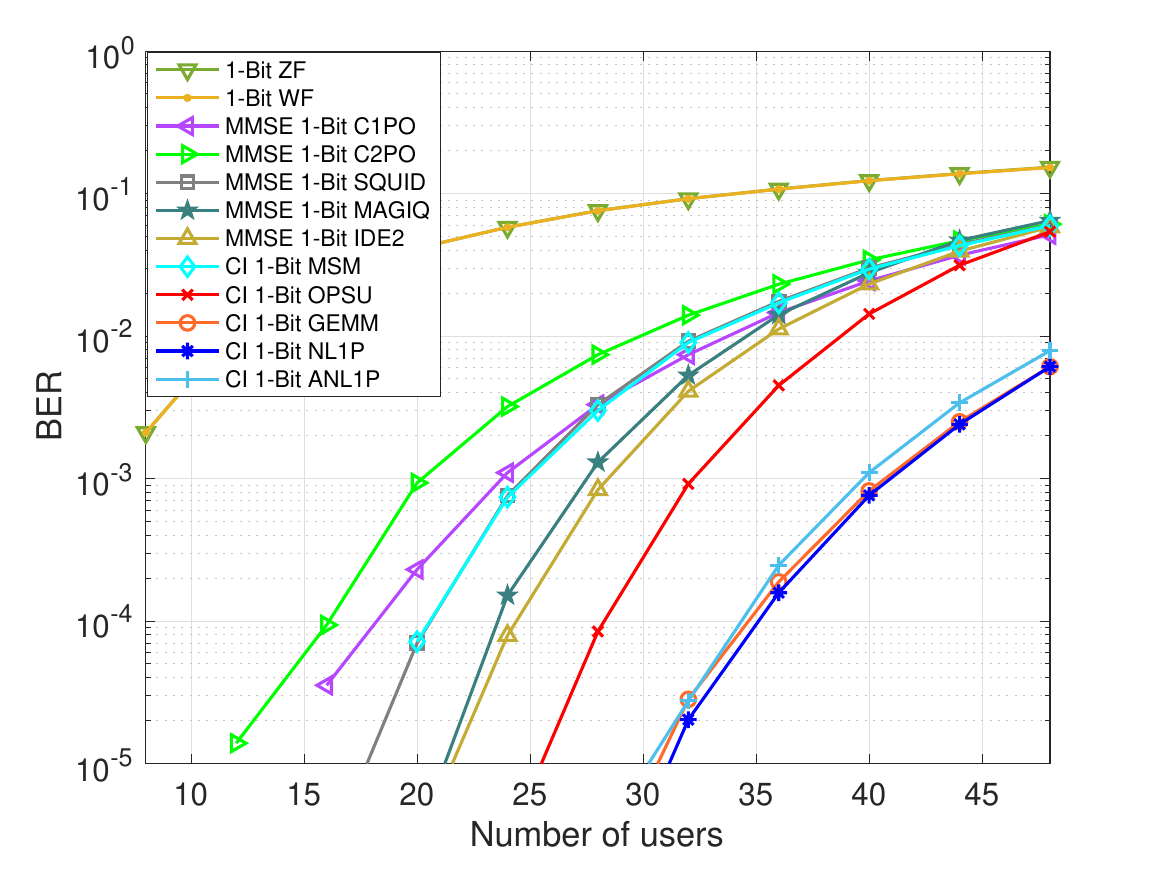}
\vspace{-15pt}
\caption{BER performance versus the number of users, where $N_t=128, M=8,$ and $\text{SNR}=20$.}
\label{fixk}
\end{figure}

In Fig. \ref{32_128_8} and Fig. \ref{16_128_16}, we investigate the more difficult cases, i.e., higher user-antenna ratio and higher-level modulation, respectively. More specifically, in Fig. \ref{32_128_8} we present the BER result for a $32\times 128$ system with $8$-PSK modulation and in Fig. \ref{16_128_16} we consider a $16\times 128$ system as in Fig. \ref{16_128_8} but with higher-order $16$-PSK modulation. The P-BB approach is not included in Fig. \ref{32_128_8} due to its prohibitively high complexity. Since the problem becomes more difficult in these two cases, it is not surprising to observe remarkable performance loss for all the precoding schemes.  In particular, only the CI-based OPSU, P-BB, {\color{black}GEMM,} and the two proposed  approaches can achieve satisfactory BER performance, while all the other approaches suffer from severe error rate floors at relatively high SNRs.   {\color{black} For both of these cases,  the two proposed algorithms and the GEMM algorithm show similar performance, and their performance advantage over the OPSU algorithm becomes more prominent than that shown in Fig. \ref{16_128_8}. }
In particular, we can observe an SNR gain up to nearly 6dB and 2.5dB  in Fig. \ref{32_128_8} and Fig. \ref{16_128_16} respectively when the BER is $10^{-4}$; as the BER becomes lower, the performance gain in terms of the SNR also becomes larger.

In Fig. \ref{fixk}, we further depict the BER of the compared one-bit precoders versus the number of users, where the number of transmit antennas at the BS is fixed as $N_t=128$, the SNR is fixed as $20$, and $8$-PSK modulation is adopted.  Among all the presented precoders,  the proposed NL1P approach achieves the best BER performance, {\color{black}followed by the GEMM algorithm and then the proposed ANL1P approach, with only  slight performance losses observed}. {\color{black}All of these three precoders} exhibit significantly better performance than the other precoding schemes in the sense that with the same BER requirement, they can serve much  more users. For example, if we require the BER to be less than $10^{-3}$, these three precoders can serve about $40$ users, while the state-of-the-art OPSU precoder can serve only $32$ users.

\begin{figure*}[!t]
\centering
\includegraphics[scale=0.45]{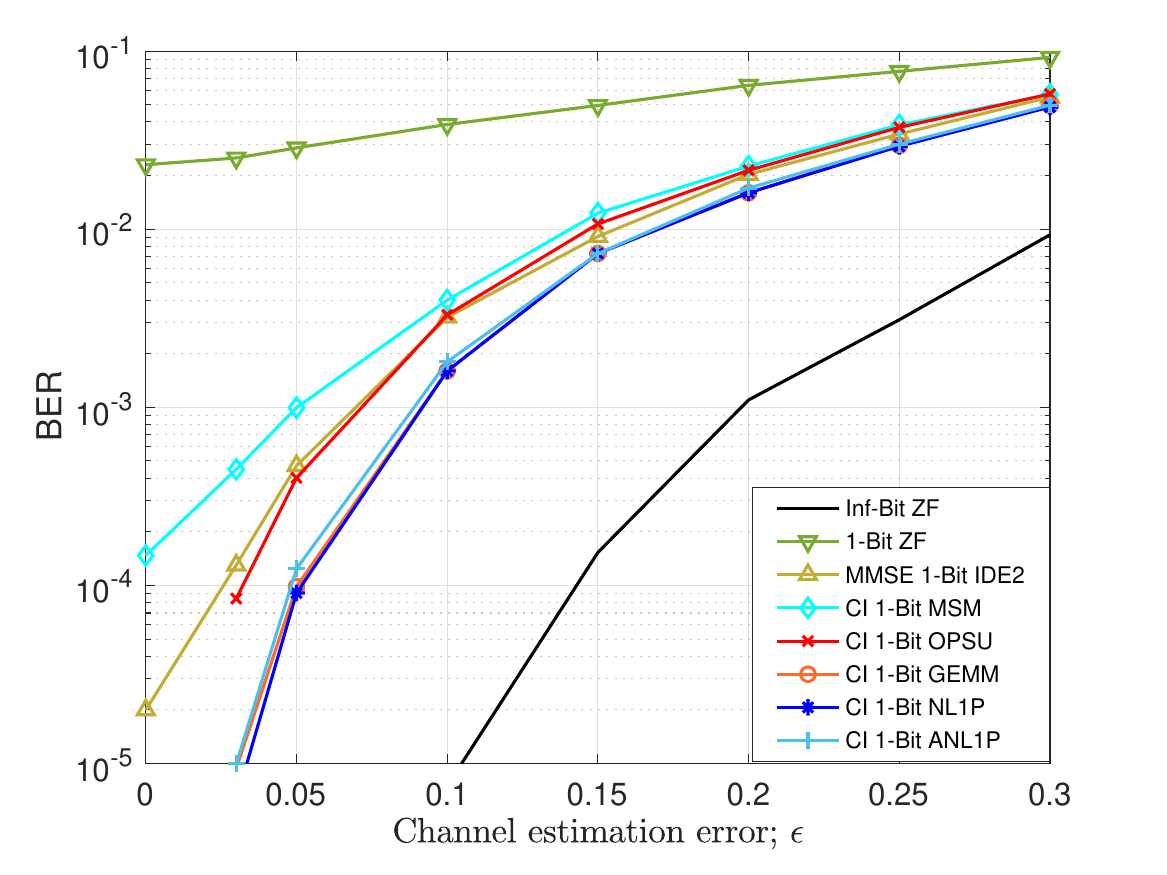}
\vspace{-20pt}
\caption{BER performance versus channel estimation error $\epsilon$, where $(K,N_t,M)=(32,128,4)$ and $\text{SNR}=15$.}
\centering
\label{CSIerror}
\vspace{-0.5cm}
\end{figure*}
{\color{black}In Fig. \ref{CSIerror}, we investigate the impact of the CSI error on the proposed approaches. The channel  accessible at the BS is modeled as 
$$\widehat{\H}=\sqrt{1-\epsilon}\H+\sqrt{\epsilon}\mathbf{Z},$$ where $\H$ is the true channel matrix, $\mathbf{Z}\sim\mathcal{CN}\left(\mathbf{0},\mathbf{I}\right)$ is independent of $\H$, and $\epsilon$ is the estimation error.  For a concise presentation, we only include the MMSE-based algorithm with the best BER performance, i.e., the IDE2 algorithm, in the figure. The P-BB approach is not included due to its prohibitively high complexity.
As can be observed, the proposed approaches still exhibit similar BER performance to GEMM and
superior performance to the others in this more practical case with CSI errors. }
\vspace{-0.48cm}
\subsection{Computational Efficiency}\label{cputime}
Now we evaluate the computational efficiency of the compared algorithms by reporting their CPU time. Since linear and MMSE-based approaches fail to achieve satisfactory BER performance in many cases, we are mostly interested in the CPU time comparison of the CI-based methods in this subsection.
In Fig. \ref{time_antenna} and Fig. \ref{time_user},  we compare the average CPU time (in seconds) of the {\color{black} CI-based} algorithms versus different numbers of transmit antennas and different numbers of users, respectively. We can make the following observations from Fig. \ref{time}. 
For the two proposed algorithms,  ANL1P runs faster than NL1P as expected.  Both of them are  generally more efficient than the LP relaxation based approaches, i.e., MSM, OPSU, and  P-BB, especially when the system scales up. More specifically, the computational costs of the MSM precoder and the OPSU precoder increase rapidly with the scale of the system, while that of our proposed approaches grow much slower. The P-BB algorithm, though with the best BER performance, is much more computationally expensive than the other methods. Its computational cost becomes prohibitively high when the number of users is large, since the complexity of the branch-and-bound procedure grows exponentially with respect to the number of users \cite{PBB}. {\color{black}In addition, the two proposed approaches are also much faster than GEMM. Specifically, the NL1P and ANL1P algorithms require approximately one-third and one-fifth of the CPU time, respectively, when compared to the GEMM algorithm. }
\begin{figure*}[!t]
\subfigure[CPU time versus the number of transmit antennas.]{
\centering
\includegraphics[scale=0.402]{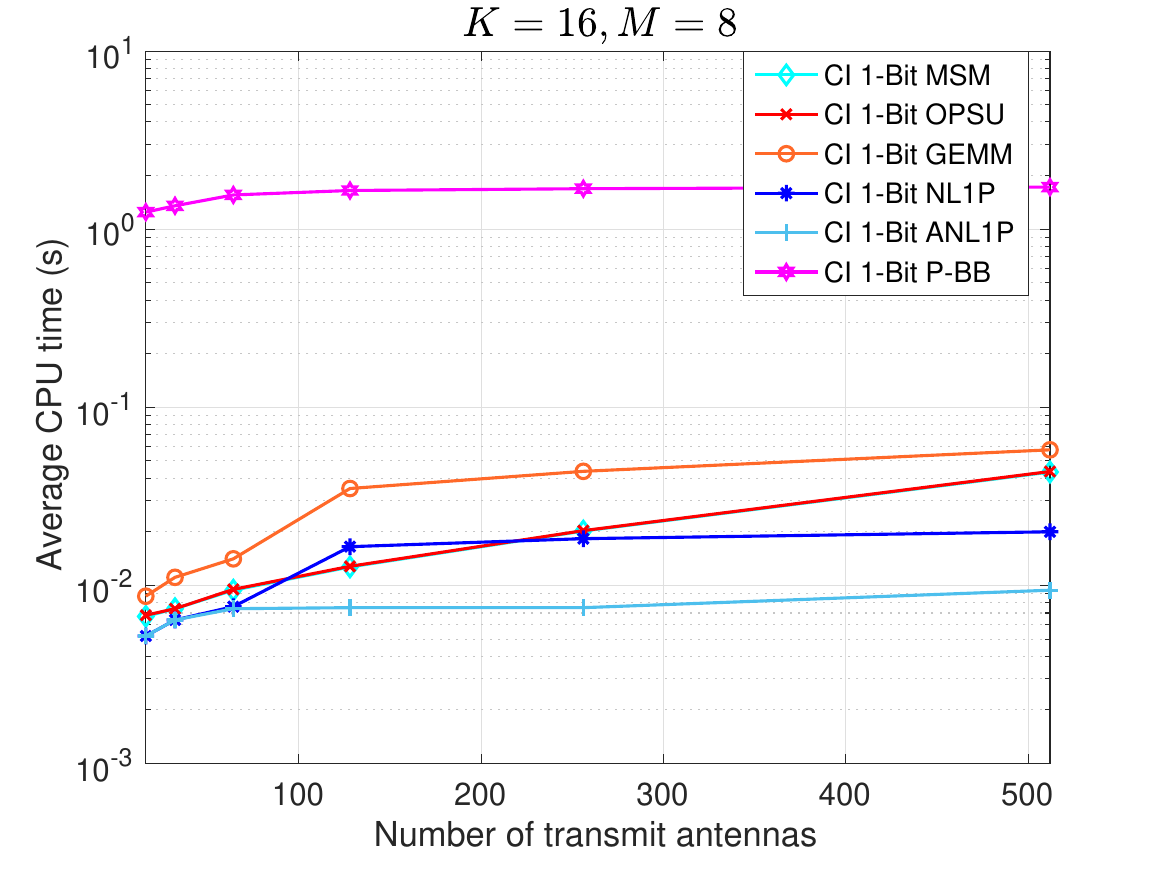}
\vspace{-10pt}
\label{time_antenna}}
\subfigure[CPU time versus the number of users.]{\includegraphics[scale=0.402]{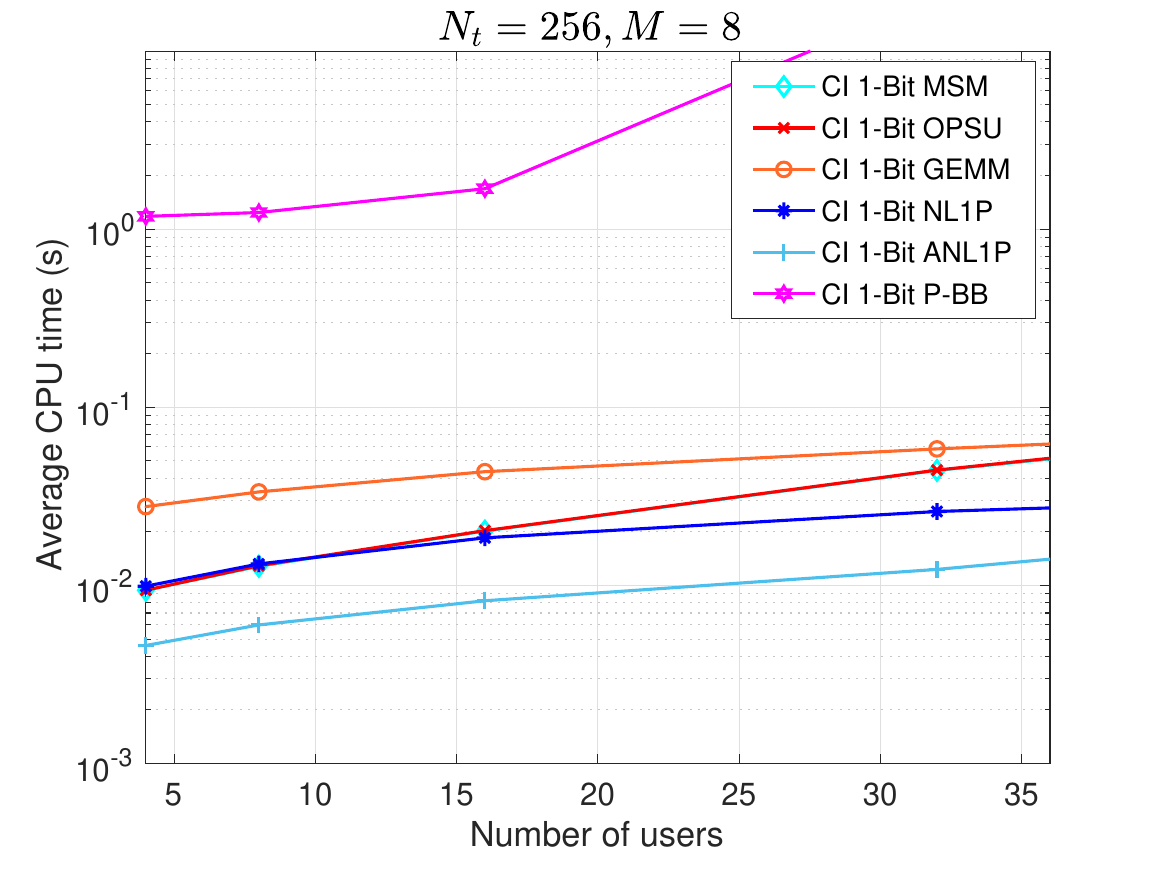}
\label{time_user}}
\vspace{-0.15cm}
\caption{CPU time versus the number of transmit antennas and the number of users.}
\label{time}
\end{figure*}

Now we can conclude this section by claiming that our proposed approaches outperform the state-of-the-art CI-based approaches in terms of both the BER performance and the computational efficiency.

\section{Conclusion}\label{conclusion}
In this paper, we considered the one-bit downlink transmission of a multiuser massive MIMO system with PSK signaling, where the CI metric is adopted. We first characterized the complexity status of the considered problem by establishing its (strong) NP-hardness. Then, we proposed a novel negative $\ell_1$ penalty (NL1P) approach for solving the considered problem, which is guaranteed to obtain a high-quality solution that satisfies the one-bit constraint. There are two main features of the NL1P approach: first, the approach is based on a novel negative $\ell_1$ penalty model, which is equivalent to the original problem both globally and locally when the penalty parameter is large enough; second, the APGDA algorithm proposed for solving the penalty model enjoys a low per-iteration complexity, which makes NL1P suitable for solving large-scale problems.  We also introduced a low-complexity implementation of the NL1P approach to further reduce its computational cost.
 Simulation results demonstrated that our approaches outperform the state-of-the-art CI-based approaches in terms of both the BER performance and the computational efficiency. 


%
\appendices
\vspace{-0.2cm}

\section{Proof of Theorem \ref{th1}}\label{AppendixA}
Notice that when $K=1$, (P$_0$) reduces to the following problem:
\begin{equation}\label{2partition1}
\begin{aligned}
\max_{\x_T}~&\min~\left\{\alpha^A,\alpha^B\right\}\\
\qquad\text{s.t. }~&\h^\mathsf{T}\x_T=\alpha^As^A+\alpha^Bs^B,\\
~&\x_T(i)\in\left\{\pm1 \pm j\right\},\quad i=1,2,\dots, N_t.
\end{aligned}
\end{equation}
Next we shall build a polynomial-time transformation from the partition problem \cite{np} to problem \eqref{2partition1}. The partition problem is to determine whether a given set of $N$ positive integers $\left\{a_1, a_2,\dots,a_N\right\}$ can be partitioned into two subsets such that the sum of elements in each subset is the same. 

Now we construct an instance of problem \eqref{2partition1} based on the given instance of the partition problem. Let the number of antennas at the BS be $N$ and the transmitted data symbol be $s=1$, which is drawn from the QPSK constellation set. In this case, $s^A=\frac{\sqrt{2}}{2}(1-j)$ and  $s^B=\frac{\sqrt{2}}{2}(1+j)$. Moreover, set the channel vector $\h$ to be $\h=\sqrt{2}\mathbf{a}$ with $\mathbf{a}=[a_1, a_2,\dots, a_N]^\mathsf{T}$. With the above constructed parameters, problem (\ref{2partition1}) can be expressed as
\begin{equation}\label{2partition2}
\begin{aligned}
\max_{\x_T}~&\min\left\{\alpha^A,\alpha^B\right\}\\
\text{s.t. }~&\left[\begin{matrix}
\alpha^A\\\alpha^B\end{matrix}\right]=\left[
\begin{matrix}
\mathbf{a}^\mathsf{T}&\mathbf{-a}^\mathsf{T}\\\mathbf{a}^\mathsf{T}&\mathbf{a}^\mathsf{T}
\end{matrix}\right]
\left[\begin{matrix}\RR(\x_T)\\\I(\x_T)\end{matrix}\right],\\
~&\x_T(i)\in\left\{\pm1\pm j\right\}, ~i=1,2,\dots,N.
\end{aligned}
\end{equation}
Let the optimal solution of problem \eqref{2partition2} be $\x_T^*.$ Since $\mathbf{a}>\mathbf{0}$, it is easy to argue that $\RR(\x_T^*)=\mathbf{1}$. By defining $\mathcal{S}=\left\{i\in\{1,2,\dots,N\}\mid\I(\x_T^*(i))=1\right\}$, it then follows that
$$\alpha^A=2\sum_{i\notin S}a_i,\quad \alpha^B=2\sum_{i\in S} a_i.$$
Now, it is straightforward to argue that the optimal value of our constructed problem (\ref{2partition2}) is $\sum_{i=1}^N a_i$ if and only if the partition problem has a ``yes'' answer. Finally, the above transformation can be done in polynomial time. Since the partition problem is NP-complete,  problem \eqref{2partition1} is NP-hard.

\section{Proof of Theorem \ref{th2}}\label{AppendixB}
The proof is based on a polynomial-time transformation from the 3-SAT problem \cite{np} to problem (P$_0$). The 3-SAT problem is to determine whether a given set of disjunctive clauses, each consisting of 3 Boolean variables,  is satisfiable. Given any instance of the 3-SAT problem consisting of $m$ disjunctive clauses $c_1,c_2,\dots, c_m$ defined on $n$ Boolean variables $x_1,x_2,\dots, x_n$, we construct below a problem instance of (P$_0$) with $K=m$ and $N_t=n+1$. 

We first express $c_k$ as $c_k=\alpha_{\pi(k)}\vee\beta_{\rho(k)}\vee\gamma_{\tau(k)}$ and define the channel vectors as {\color{black}$\h_k=e^{-\frac{j\pi}{4}}{\mathbf{g}}_k=e^{-\frac{j\pi}{4}}[{g_{k1}}, {g_{k2}},\dots,{g_{kn}},1]^\mathsf{T}\in \R^{n+1},~k=1,2,\dots,m$, with}
\begin{equation*}
g_{ki}=\left\{
\begin{aligned}
&1,\qquad\text{if}~\alpha_{\pi(k)}=x_i~ \text{or}~\beta_{\rho(k)}=x_i~\text{or}~\gamma_{\tau(k)}=x_i;\\
-&1,\qquad\text{if}~\alpha_{\pi(k)}=\bar{x}_i~ \text{or}~\beta_{\rho(k)}=\bar{x}_i~\text{or}~\gamma_{\tau(k)}=\bar{x}_i;\\
&0,\qquad\text{otherwise}, \hspace{5.8cm}
\end{aligned}
\right. i=1,2,\dots, n.
\end{equation*}
For example, if $c_k=x_1\vee\bar{x}_3\vee x_4,$ then $\mathbf{g}_k=[1,0,-1,1,0,\dots,0,1]$. Moreover, we set the modulation scheme to be {\color{black}QPSK and the data symbols for all users to be $$s_k=1,\quad k=1,2,\dots, m.$$ It follows immediately that $s_k^A=e^{-\frac{j\pi}{4}}$ and $ s_k^B=e^{\frac{j\pi}{4}}$ for all $k=1,2,\dots,m$. With the above constructed parameters and by multiplying $e^{\frac{j\pi}{4}}$ on both sides of the first constraint of (P$_0$)}, problem (P$_0$) becomes
\begin{equation}\label{construct1}
\begin{aligned}
\max_{\x_T}~&\min_{k\in\{1,2,\dots,m\}}~\{\alpha_k^A,\alpha_k^B\}\\
\text{s.t. }~&\mathbf{g}_k^\mathsf{T}\Re(\x_T)=\alpha_k^A,\quad k=1,2,\dots,m,\\
~&\mathbf{g}_k^\mathsf{T}\Im(\x_T)=\alpha_k^B,\quad k=1,2,\dots,m,\\
~&\x_T(i)\in\left\{\pm1\pm j\right\},\quad i=1,2,\dots,n+1,
\end{aligned}
\end{equation}
which is equivalent to 
\begin{equation}\label{nphard1}
\begin{aligned}
\max_{\y,t}~&t\\
\text{s.t. }~&\mathbf{g}_k^\mathsf{T}\y\geq t,\quad k=1,2,\dots, m,\\
~&y_i\in\left\{-1,1\right\},\quad i=1,2,\dots n+1.
\end{aligned}
\end{equation}
Let the optimal solution of \eqref{nphard1} be $\y^*$.  Since the last entries of all $\{\mathbf{g}_k\}_{k=1}^m$ are $1$, it is easy to argue that $y^*_{n+1}=1$. Based on this,  and by further defining 
$\hat{\mathbf{g}}_k=\mathbf{g}_k[1:n]$ for all $k=1,2,\dots,m$ and $z_i=(y_i+1)/2$ for all $i=1,2,\dots, n$, problem \eqref{nphard1} can be equivalently expressed as
\begin{equation}\label{nphard2}
\begin{aligned}
\max_{{\z},t}~&t\\
\text{s.t. }~&2\hat{\mathbf{g}}_k^\mathsf{T}\z-\mathbf{1}^\mathsf{T}\hat{\mathbf{g}}_k+1\geq t,\quad k=1,2,\dots, m,\\
~&z_i\in\left\{0,1\right\},\quad i=1,2,\dots n.
\end{aligned}
\end{equation} 

Now we claim that the 3-SAT problem is satisfied if and only if the optimal value of our constructed problem is greater than or equal to zero, or equivalently, there exists $\z\in\{0,1\}^n$ such that 
\begin{equation}\label{nphard3}
2\hat{\mathbf{g}}_k^\mathsf{T}\z-\mathbf{1}^\mathsf{T}\hat{\mathbf{g}}_k+1\geq0, \quad k=1,2,\dots, m.
\end{equation}
From the choice of ${\mathbf{g}}_k$, we know that $c_k$ is satisfied if and only if 
\begin{equation*}
\begin{aligned}
g_{k\pi(k)}x_{\pi(k)}+\frac{1-{g}_{k\pi(k)}}{2}+g_{k\rho(k)}x_{\rho(k)}+\frac{1-{g}_{k\rho(k)}}{2}+g_{k\tau(k)}x_{\tau(k)}+\frac{1-{g}_{k\tau(k)}}{2}\geq 1,
\end{aligned}
\end{equation*}
which is equivalent to 
$$2\hat{\mathbf{g}}_k^\mathsf{T}\x-\mathbf{1}^\mathsf{T}\hat{\mathbf{g}}_k+1\geq 0.$$
Therefore, if there exists a truth assignment $x_1,x_2,\dots,x_n$ for the 3-SAT problem, we can simply set $z_i=x_i,~i=1,2,\dots, n$, to obtain a solution of (\ref{nphard2}) with the objective value greater than or equal to $0$. On the other hand, if the optimal value of the constructed problem is greater than or equal to $0$, i.e., there exists $\z\in\{0,1\}^n$ satisfying (\ref{nphard3}), we can simply assign $x_i=z_i,~i=1,2,\dots, n$, to obtain a truth assignment.  Since the transformation is in polynomial time and the 3-SAT problem is strongly NP-hard, we can conclude that problem (P$_0$) is strongly NP-hard. 

It follows immediately from the above proof that there is no polynomial-time constant approximation algorithm for solving (P$_0$).
Otherwise, we can check whether the optimal value of the constructed problem \eqref{construct1} is nonnegative in polynomial time, which in turn solves the corresponding 3-SAT problem in polynomial time. This contradicts with the strong NP-hardness of the 3-SAT problem.

\section{Proof of Theorem \ref{equivalence}}\label{AppendixC}
{\color{black} Notice that (i) in Theorem \ref{equivalence} is a direct corollary of (ii) in Theorem \ref{equivalence}, thus it is sufficient to prove (ii) of Theorem \ref{equivalence}.}

First, given a penalty parameter $\lambda>\max_l\|\mathbf{a}_l\|_{\infty}$ and a local minimizer $\bar{\x}$ of (P$_\lambda$), we will show that $\bar{\x}$ is a feasible point of (P). Suppose for contradiction that there exists $s\in\{1,2,\dots n\}$  such that $|\bar{x}_s|<1$.  We claim that for any $\delta>0$, there exists $\z=[z_1,z_2,\dots,z_n]^\mathsf{T}\in B(\bar{\x},\delta)\cap[-1,1]^n$ such that $$\max_l \mathbf{a}_l^\mathsf{T}\z-\lambda\|\z\|_1<\max_l \mathbf{a}_l^\mathsf{T}\bar{\x}-\lambda\|\bar{\x}\|_1,$$ which contradicts with the fact that $\bar{\x}$ is a local minimizer.
Specifically, let 
\begin{equation*}
z_i=\left\{
\begin{aligned}
&\text{sgn}(\bar{x}_i)\min\{|\bar{x}_i|+\delta,1\}, \quad \hspace{0.1cm}\text{if }i=s;\\
&\quad ~\bar{x}_i,\qquad \qquad\qquad\qquad\hspace{0.3cm}\text{otherwise},
\end{aligned}
\right.\quad~ i=1,2,\dots,n.
\end{equation*}
It is easy to check that $\z\in B(\bar{\x},\delta)\cap[-1,1]^n$ and
$$\|\z\|_1-\|\bar{\x}\|_1=\|\z-\bar{\x}\|_1=\min\{1-|\bar{x}_s|,\delta\},$$
which, together with $\lambda>\max_l\|\mathbf{a}_l\|_\infty$,  implies that
\begin{equation*}
\begin{aligned}
&\max_l\mathbf{a}_l^\mathsf{T}\z-\max_l\mathbf{a}_l^\mathsf{T}\bar{\x}-\lambda(\|\z\|_1-\|\bar{\x}\|_1)\\
\leq&\max_l\|\mathbf{a}_l\|_\infty\|\z-\bar{\x}\|_1-\lambda(\|\z\|_1-\|\bar{\x}\|_1)\\
=&(\max_l\|\mathbf{a}_l\|_\infty-\lambda)\min\{1-|\bar{x}_s|, \delta\}<0.
\end{aligned}
\end{equation*}
Therefore,  we can conclude that any local minimizer of (P$_\lambda$) is a feasible point of (P) when $\lambda>\max_l\|\mathbf{a}_l\|_{\infty}$.

On the other hand, let $\bar{\x}$ be a feasible point of (P), and we will show that for any  $\lambda>\max_l\|\mathbf{a}_l\|_\infty$, it is also a (strict) local minimizer of (P$_\lambda$).  The main idea of the  proof is similar to that in \cite[Theorem 3.5]{penaltyproof}.
 We first define $\x^s,~s=1,2,\dots, 2^n-1$, to be all the remaining feasible points of (P) except $\bar{\x}$ and set $\d^s=\x^s-\bar{\x}$ for all $s$. 
Given $t\in(0,\frac{1}{2})$, let $\x_s=\bar{\x}+t\d^s, s=1,2,\dots, 2^n-1,$ and let $\x_N=\bar{\x}$, where $N=2^n$. We claim that ${\x}_N$ is the local minimizer in $\text{conv}(\x_1,\x_2,\dots,\x_N)$, where $\text{conv}(\cdot)$ denotes the convex hull of the corresponding set.

Note that any given $\x\in \text{conv}(\x_1,\x_2,\dots,\x_N)$  can be expressed as 
$\x=\sum_{s=1}^N\mu_s\x_s$ with some $\boldsymbol{\mu}=[\mu_1,\mu_2,\dots,\mu_N]\in\R^N$ satisfying $\sum_{s=1}^N\mu_s=1$ and $\boldsymbol{\mu}\geq\mathbf{0}.$
It follows that
\begin{equation}\label{17}
\begin{aligned}
\max_l\mathbf{a}_l^\mathsf{T}\x-\max_l\mathbf{a}_l^\mathsf{T}\x_N=&\max_l\mathbf{a}_l^\mathsf{T}\left(\sum_{s=1}^N\mu_s\x_s\right)-\max_l\mathbf{a}_l^\mathsf{T}\x_N\\
\geq &\mathbf{a}_{l_0}^\mathsf{T}\left(\sum_{s=1}^N\mu_s\x_s-\x_N\right)~\quad\qquad\\
\geq &-\|\mathbf{a}_{l_0}\|_\infty\sum_{s=1}^N\mu_s\|\x_s-\x_N\|_1\\
\geq&-\max_l\|\mathbf{a}_l\|_\infty\sum_{s=1}^{N-1}\mu_st\|\d^s\|_1,
\end{aligned}
\end{equation}
where $ l_0\in\arg\max_l\mathbf{a}_l^\mathsf{T}\x_N$.
Furthermore, 
\begin{equation}\label{18}
\begin{aligned}
\|\x_N\|_1-\|\x\|_1&=\|\x_N\|_1-\left\|\sum_{s=1}^N\mu_s\x_s\right\|_1\\
&\geq \|\x_N\|_1-\sum_{s=1}^N\mu_s\|\x_s\|_1\\
&=\sum_{s=1}^{N-1}\mu_s\left(\|\x_N\|_1-\|\x_s\|_1\right).
\end{aligned}
\end{equation}

Combining (\ref{17}) and (\ref{18}) gives
\begin{equation}\label{19}
\begin{aligned}
&\max_l\mathbf{a}_l^\mathsf{T}\x-\lambda\|\x\|_1-\left(\max_l\mathbf{a}_l^\mathsf{T}\x_N-\lambda\|\x_N\|_1\right)\\
\geq &\sum_{s=1}^{N-1}\mu_s\left(\lambda\|\x_N\|_1-\lambda\|\x_s\|_1-\max_l\|\mathbf{a}_l\|_\infty t\|\d^s\|_1\right).
\end{aligned}
\end{equation}
For all $s\in\{1,2,\dots, N-1\}$, we denote
\begin{equation*}
\begin{aligned}
&\Gamma_1^s=\left\{i\mid\bar{\x}(i)=1, \x^s(i)=1\right\},~~\ \Gamma_2^s=\left\{i\mid\bar{\x}(i)=1, \x^s(i)=-1\right\},\\
&\Gamma_3^s=\left\{i\mid\bar{\x}(i)=-1, \x^s(i)=1\right\},~\hspace{-0.02cm} \Gamma_4^s=\left\{i\mid\bar{\x}(i)=-1, \x^s(i)=-1\right\}.
\end{aligned}
\end{equation*}
Since $t\in(0,\frac{1}{2})$, we have 
\begin{equation*}
\begin{aligned}
\|\x_N\|_1-\|\x_s\|_1&=|\Gamma_2^s|(1-|1-2t|)+|\Gamma_3^s|(1-|-1+2t|)
=(|\Gamma_2^s|+|\Gamma_3^s|)2t
\end{aligned}
\end{equation*}
and 
$$\|\d^s\|_1=2(|\Gamma_2^s|+|\Gamma_3^s|),\quad s=1,2,\dots, N-1,$$
which, together with \eqref{19}, implies
\begin{equation*}
\begin{aligned}
&\max_{l}\mathbf{a}_l^\mathsf{T}\x-\lambda\|\x\|_1-(\max_l\mathbf{a}_l^\mathsf{T}\x_N-\lambda\|\x_N\|_1)\geq 2\sum_{s=1}^{N-1}t\mu_s(|\Gamma_2^s|+|\Gamma_3^s|)(\lambda-\max_l\|\mathbf{a}_l\|_\infty)\geq0,
\end{aligned}
\end{equation*}
{\color{black}where the last inequality holds strictly if $\mu_1, \mu_2, \cdots, \mu_{N-1}$ are not all $0$, i.e., $\x\neq \x_N$.}
Therefore, for any $t\in (0,\frac{1}{2}) $ and any $\lambda>\max\|\mathbf{a}_l\|_\infty$, it holds that for all  $\x\in\text{conv}(\x_1,\x_2,\dots, \x_N)$ {\color{black} and $\x\neq \x_N$},
$$\max_l\mathbf{a}_l^\mathsf{T}\x-\lambda\|\x\|_1{\color{black}>}\max_l\mathbf{a}_l^\mathsf{T}\x_N-\lambda\|\x_N\|_1,$$ which proves our claim.
Moreover, we can always choose a sufficiently small but fixed $\epsilon>0$ such that $B(\x_N, \epsilon)\cap [-1,1]^n\subset \text{conv}(\x_1,\x_2,\dots, \x_N)$. Consequently, $\bar{\x}=\x_N$ is a {\color{black}(strict)} local minimizer of (P$_\lambda$), which completes our proof.

\section{Proof of Theorem \ref{convergethe}}\label{AppendixD}
In this section, we give the proof of Theorem \ref{convergethe}. We first give an auxiliary lemma that is important to the proof in Appendix \ref{auxilarylemma} and then give the main proof in Appendix \ref{mainproof}. 
\subsection{An Auxiliary Lemma}\label{auxilarylemma}
\begin{lemma}
Suppose that Assumption \ref{assumption1} holds and assume that $\{c_k\}$ is a nonnegative monotonically decreasing sequence. Let $\{(\x_k,\y_k)\}$ be the sequence generated by Algorithm \ref{algorithm1} with $\rho_k=\rho$. Also  denote  $F_{k+1}=F(\x_{k+1},\y_{k+1})$,
\begin{equation}\label{S}
S_{k+1}=\frac{4}{\rho^2c_{k+1}}\|\y_{k+1}-\y_k\|_2^2-\frac{4}{\rho}\left(\frac{c_{k-1}}{c_k}-1\right)\|\y_{k+1}\|_2^2,
\end{equation}
\begin{equation}\label{Phi}
\Phi_{k+1}=F_{k+1}+S_{k+1}-\frac{7}{2\rho}\|\y_{k+1}-\y_k\|_2^2-\frac{c_k}{2}\|\y_{k+1}\|_2^2.
\end{equation}
If 
\begin{equation}\label{conditionofcrho}
\frac{1}{c_{k+1}}-\frac{1}{c_k}\leq \frac{\rho}{10}, \quad\rho\leq \frac{2}{L_y+2c_1},
\end{equation}
then for all $k\geq 1$, it holds that
\begin{equation*}
\begin{array}{l}
 \begin{aligned}
\Phi_{k+1}-\Phi_{k} \leq &-\left(\frac{\tau_k-L_x}{2}-\frac{{\rho} L_{12}^{2}}{2}-\frac{8 L_{12}^{2}}{{\rho} c_{k}^{2}}\right)\left\|\x_{k+1}-\x_{k}\right\|_2^{2}\\
&-\frac{1}{10 {\rho}}\left\|\y_{k+1}-\y_{k}\right\|_2^{2}+\frac{c_{k-1}-c_{k}}{2}\left\|\y_{k+1}\right\|_2^{2}\\
&+\frac{4}{{\rho}}\left(\frac{c_{k-2}}{c_{k-1}}-\frac{c_{k-1}}{c_{k}}\right)\left\|\y_{k}\right\|_2^{2}.
\end{aligned}
\end{array}
\end{equation*}
\end{lemma}

\begin{proof}
Since $\nabla_{\x} f(\x,\y)$ is Lipschitz continuous for fixed $\y$, we have
\begin{equation*}
\begin{aligned}
f(\x_{k+1},\y_k)-f(\x_k,\y_k)\leq& \left<\nabla_{\x} f(\x_k,\y_k),\x_{k+1}-\x_k\right>+\frac{L_x}{2}\|\x_{k+1}-\x_k\|_2^2.
\end{aligned}
\end{equation*}
By the update rule of $\x$, we further have
\begin{equation*}
\begin{aligned}
\left<\nabla_{\x} f(\x_k,\y_k),\x_{k+1}-\x_k\right>-g(\x_{k+1})+g(\x_k)
\leq-\frac{\tau_k}{2}\|\x_{k+1}-\x_k\|_2^2.
\end{aligned}
\end{equation*}
Combining the above two inequalities yields 
$$ F(\x_{k+1},\y_k)-F(\x_k,\y_k)\leq \frac{L_x-\tau_k}{2}\|\x_{k+1}-\x_k\|_2^2.$$
Then the rest of the proof is the same as in \cite[Lemma 3.6]{xu2020unified}.
\end{proof}
\subsection{Proof of Theorem \ref{convergethe}}\label{mainproof}
Now we are ready to give the proof of Theorem \ref{convergethe}. With the selected parameters in the theorem, it is easy to check that $\rho\leq \frac{2}{L_y+2c_1}$ and $\lim_{k\to\infty}\left(\frac{1}{c_{k+1}}-\frac{1}{c_k}\right)=0,$ and thus there exists $k_0$ such that condition \eqref{conditionofcrho} holds for all $k\geq k_0$.

In the following, we shall first prove that 
$$\tau_k\|\x_{k+1}-\x_k\|_2\rightarrow 0~\text{and}~\|\y_{k+1}-\y_k\|_2\rightarrow 0.$$
Let $\alpha_k=\frac{8(\beta_2-1)L_{12}^2}{\rho c_k^2},$ and then $\tau_k$ can be expressed as   $ \tau_k=\frac{16L_{12}^2}{\rho c_k^2}+2\alpha_k+\beta_3.$ Since $\beta_3\geq L_x+\rho L_{12}^2$, it follows from the above lemma that for all $k\geq k_0$,
\begin{equation}\label{32}
\begin{aligned}
&\alpha_k\|\x_{k+1}-\x_k\|_2^2+\frac{1}{10\rho}\|\y_{k+1}-\y_k\|_2^2\\
\leq \ &\Phi_{k}-\Phi_{k+1}+\frac{4}{{\rho}}\left(\frac{c_{k-2}}{c_{k-1}}-\frac{c_{k-1}}{c_{k}}\right)\left\|\y_{k}\right\|_2^{2}+\frac{c_{k-1}-c_{k}}{2}\left\|\y_{k+1}\right\|_2^{2}.
\end{aligned}
\end{equation}
For all $K>k_0$, summing both sides of (\ref{32}) from $k=k_0$ to $k=K$ gives 
 \begin{equation*}
 \begin{aligned}
 &\sum_{k=k_0}^K\alpha_k\|\x_{k+1}-\x_k\|_2^2+\frac{1}{10\rho}\|\y_{k+1}-\y_k\|_2^2 \\
 \leq\ &\Phi_{k_0}-\Phi_{K+1}+\frac{4}{\rho}\left(\frac{c_{{k_0}-2}}{c_{k_0-1}}-\frac{c_{K-1}}{c_K}\right)\sigma_y^2+\frac{c_{k_0-1}-c_K}{2}\sigma_y^2,\\
 \end{aligned}
 \end{equation*}
 where $\sigma_y=\max\{\|\y\|_2~|~\y\in \mathcal{Y}\}$.
  Furthermore, it follows from the definitions \eqref{S} and \eqref{Phi} that $\Phi_k$ is bounded from below, and thus we have 
 \begin{equation}\label{xto0}
 \sum_{k=1}^{\infty} \alpha_k\|\x_{k+1}-\x_k\|_2^2<+\infty
 \end{equation}
 and
 \begin{equation*}
  \sum_{k=1}^{\infty} \|\y_{k+1}-\y_k\|_2^2<+\infty,
 \end{equation*}
which immediately shows that  $$ \|\y_{k+1}-\y_k\|_2\rightarrow 0.$$
On the other hand, since $\lim_{k\to\infty}c_k=0$ and $\beta_2>1$, we have 
\begin{equation*}
\begin{aligned}
\lim_{k\to\infty}\frac{\tau_k}{\alpha_k}&=\lim_{k\to\infty}\frac{\frac{16\beta_2 L_{12}^2}{\rho c_k^2}+\beta_3}{\frac{8(\beta_2-1)L_{12}^2}{\rho c_k^2}}=\frac{2\beta_2}{\beta_2-1},\\
\end{aligned}
\end{equation*}
which implies that the sequence $\left\{\frac{\tau_k}{\alpha_k}\right\}$ is bounded, i.e., there exists $d_1$ such that 
$$\frac{\tau_k}{\alpha_k}\leq d_1,\quad\forall k\geq 1.$$
Therefore, the following inequality holds for all $k\geq 1$,
\begin{equation*}
\frac{1}{d_1\tau_k}\tau_k^2\|\x_{k+1}-\x_k\|_2^2\leq \alpha_k\|\x_{k+1}-\x_k\|_2^2,
\end{equation*}
which, together with \eqref{xto0}, gives
\begin{equation*}
\sum_{k=1}^\infty\frac{1}{d_1\tau_k}\tau_k^2\|\x_{k+1}-\x_k\|_2^2<+\infty.
\end{equation*}
Note that $d_1\tau_k\leq d_1^2\alpha_k=\mathcal{O}(k^{2\gamma})$ with $0<\gamma\leq 0.5$. Then we know  from  the above assertion that 
\begin{equation}\label{taudiffx}
\tau_k\|\x_{k+1}-\x_k\|_2\rightarrow 0.
\end{equation}

Next we shall show that any limit point of $\{(\x_k,\y_k)\}$ is a stationary point of the corresponding min-max problem. Given a limit point $(\hat{\x},\hat{\y})$ of $\{(\x_k,\y_k)\}$, there exists $\{k_j\}$ such that 
$$\lim_{j\to\infty}(\x_{k_j}, \y_{k_j})=(\hat{\x},\hat{\y}).$$
By the update rules of $\x$ and $\y$ in Algorithm \ref{algorithm1}, it holds that
\begin{subequations}
\begin{align}
\mathbf{0}&\in\nabla_{\x} f(\x_{k_j},\y_{k_j})-\partial g(\x_{k_j+1})+\tau_{k_j}(\x_{k_j+1}-\x_{k_j})+\partial\mathbb{I}_\mathcal{X}(\x_{k_j+1}),\label{xa}\\
\mathbf{0}&\in-\nabla_{\y} f(\x_{k_j+1},\y_{k_j})+\frac{1}{\rho}(\y_{k_j+1}-\y_{k_j})+c_{k_j}\y_{k_j}+\partial\mathbb{I}_{\mathcal{Y}}(\y_{k_j+1}).
\end{align}
\end{subequations}
The above (\ref{xa}) can be equivalently expressed as
\begin{equation}\label{limit}
\begin{aligned}
\langle-\nabla_{\x}f(\x_{k_j},\y_{k_j})+\mathbf{s}_{k_{j}+1}-\tau_{k_j}(\x_{k_j+1}-\x_{k_j}),\x-\x_{k_j+1}\rangle\leq0,\quad\forall \x\in\mathcal{X},
\end{aligned}
\end{equation}
where $\mathbf{s}_{k_j+1}$ is an element in $\partial g(\x_{k_j+1})$ that guarantees (\ref{xa}) holds. By Assumption \ref{assumptionong}, we know that the sequence $\{\mathbf{s}_{k_j+1}\}$ is bounded. Without loss of generality, we assume 
$$\lim_{j\to\infty}\mathbf{s}_{k_j+1}=\hat{\mathbf{s}},$$
otherwise we can extract a convergent subsequence.
With \eqref{taudiffx} and notice that $\tau_k\rightarrow \infty$, we have $\x_{k_j+1}\rightarrow \hat{\x}$.
Taking limits of the left hand side of inequality (\ref{limit}) gives 
$$\left<-\nabla_{\x} f(\hat{\x},\hat{\y})+\hat{\mathbf{s}},\x-\hat{\x}\right>\leq 0,\quad \forall \x\in\mathcal{X},$$
which further implies that
\begin{equation}\label{xstationary}
0\in\nabla_{\x} f(\hat{\x},\hat{\y})-\hat{\mathbf{s}}+\partial\mathbf{1}_{\mathcal{X}}(\hat{\x})\subset \nabla_{\x} f(\hat{\x},\hat{\y})-\partial g(\hat{\x})+\partial \mathbf{1}_\mathcal{X}(\hat{\x}).
\end{equation}
The last inclusion holds since $g$ is a proper closed convex function, and thus the graph of $\partial g(\x)$ is closed\cite[Theorem 24.4]{convexanalysis}, i.e., $\mathbf{s}_{k_j+1}\in\partial g(\x_{k_j+1})$ with $\mathbf{s}_{k_j+1}\rightarrow\hat{\mathbf{s}}$ and  $\{\x_{k_j+1}\}\rightarrow \hat{\x}$ can imply $\hat{\mathbf{s}}\in\partial g(\hat{\x})$.

Since $c_k\rightarrow 0$ and $\|\y_{k+1}-\y_k\|_2\rightarrow 0$, similarly we can show that 
\begin{equation}\label{ystationary}
\mathbf{0}\in-\nabla_{\y} f(\hat{\x},\hat{\y})+\partial\mathbb{I}_{\mathcal{Y}}(\hat{\y}).
\end{equation}
Combining (\ref{xstationary}) and (\ref{ystationary}), we can conclude that $(\hat{\x},\hat{\y})$ is a stationary point.

\section{Derivation of Solutions to \eqref{distributedform}}\label{AppendixE}
We first notationally simplify \eqref{distributedform} as
\begin{equation}\label{px}
x^*=\arg\min_{-1\leq x\leq 1}~(x-a)^2+b|x|,
\end{equation}
 where $a=\x_{k}(i)-\frac{\A_i^\mathsf{T}\y_k}{\tau_k}$ and $b=-\frac{2\lambda}{\tau_k}<0$ in our problem.  Next we shall consider the two cases of $a\neq 0$ and $a=0$ separately.

When $ a\neq 0$, it is easy to argue that the optimal solution $x^*$ of \eqref{px} must be of the form $x^*=\text{sgn}(a)r$, where $r$ is some nonnegative number in $[0,1]$,  and thus  the optimization problem in \eqref{px} can be equivalently expressed as 
\vspace{-0.1cm}\begin{equation}\label{px2}
\min_{0\leq r\leq 1}\left(r-\left(|a|-\frac{b}{2}\right)\right)^2.
\end{equation}
 Since in our problem $|a|-\frac{b}{2}>-\frac{b}{2}>0$, the optimal solution of \eqref{px2} is $r^*=\min\left\{|a|-\frac{b}{2},1\right\}$, and thus the optimal solution of the original problem \eqref{px} is $x^*=\text{sgn}(a)\min\left\{|a|-\frac{b}{2},1\right\}$.

 When $a=0$, the optimization problem in \eqref{px} becomes
 $\min_{|x|\leq 1}|x|^2+b|x|,$ whose optimal solution is  $|x^*|=\min\left\{-\frac{b}{2},1\right\}$, or equivalently, $x^*\in\left\{\min\left\{-\frac{b}{2},1\right\},-\min\left\{-\frac{b}{2},1\right\}\right\}.$

 Combining the above discussions, we can conclude that the optimal solution of \eqref{px} is given by
 \begin{equation*}
x^*=\left\{
\begin{aligned}
&\text{sgn}(a)\min\left\{|a|-\frac{b}{2},1\right\},~~~\text{if}~a\neq 0;\\
&\pm \min\left\{-\frac{b}{2},1\right\},\hspace{1.35cm}~~~\text{if}~a=0.
\end{aligned}\right.
\end{equation*}

 \section{Proof of Lemma \ref{criticalthe}}\label{AppendixF}
Given a stationary point $\hat{\x}$ of (P$_\lambda$), there exists $\u,\v\in\R_+^n$ such that 
\begin{subequations}
\begin{align}
&\mathbf{0}\in\partial(\max_l\mathbf{a}_l^\mathsf{T}\hat{\x})-\lambda\partial\|\hat{\x}\|_1-\u+\v,\label{kkt1}\\
&u_i(\hat{x}_i+1)=0,~v_i(\hat{x}_i-1)=0,~i=1,2,\dots,n.\label{kkt2}
\end{align}
\end{subequations}
Next we shall show that if $\lambda>\max_l\|\mathbf{a}_l\|_\infty$, $\hat{\x}$ must satisfy $|\hat{x}_i|\in\{0,1\}$ for all $i=1,2, \dots,n.$ Suppose for contradiction that there exists $s$, such that $0<|\hat{x}_s|<1$. It follows immediately from \eqref{kkt2} that $u_s=v_s=0$. Moreover, from the calculation rule of the subdifferential\cite{convexanalysis}, we know that $(\partial\|\hat{\x}\|_1)_s=\text{sgn}(\hat{x}_s)$ and 
$$\partial (\max_l\mathbf{a}_l^\mathsf{T}\hat{\x})=\left\{\A^\mathsf{T} \mathbf{t}\mid \mathbf{t}\in\Delta, 
 t_i=0~\text{if}~ i\notin \mathcal{I}\right\},$$ where $\Delta=\{\mathbf{t}\in\mathbb{R}^m\mid\mathbf{1}^\mathsf{T}\mathbf{t}=1,~\mathbf{t}\geq\mathbf{0}\}$ 
 and $\mathcal{I}$ is defined as 
$$\mathcal{I}=\left\{i\in\{1,2,\dots,m\}\mid\mathbf{a}_i^\mathsf{T}\hat{\x}=\max_l\mathbf{a}_l^\mathsf{T}\hat{\x}\right\}.$$
 Therefore, for all $\s\in \partial (\max_l\mathbf{a}_l^\mathsf{T}\hat{\x})$,  we have 
\begin{equation*}
\begin{aligned}
\|\s\|_\infty=\|\A^\mathsf{T}\mathbf{t}\|_\infty&=\max_l\left|\mathbf{a}_l^\mathsf{T}\mathbf{t}\right|\\
&\leq \max_l\|\mathbf{a}_l\|_\infty\|\mathbf{t}\|_1\\
&= \max_l\|\mathbf{a}_l\|_\infty<\lambda,
\end{aligned}
\end{equation*}
which implies that the condition (\ref{kkt1}) cannot be satisfied for the $s$-th component. As a result, $\hat{\x}$ must have all its elements being either $\pm1$ or $0$.

\section{Proof of Theorem \ref{c1}}\label{AppendixG}
From the closed-form solution \eqref{solutionx}, we know that if $\lambda>\max_{l\in\{1,2,\dots,m\}}\|\mathbf{a}_l\|_\infty$, then for all $i\in\{1,2,\dots,n\}$,
\begin{equation*}
\begin{aligned}
\left|\x_{k+1}(i)\right|=\min\left\{\left|\x_k(i)-\frac{\A_i^\mathsf{T}\y_k}{\tau_k}\right|+\frac{\lambda}{\tau_k},1\right\}&\geq \min\left\{\left|\x_k(i)\right|+\frac{\lambda-\left|\A_i^\mathsf{T}\y_k\right|}{\tau_k},1\right\}\\&\geq\min\left\{\left|\x_k(i)\right|,1\right\},
\end{aligned}
\end{equation*}
where the last inequality holds since
\begin{equation*}
\begin{aligned}
\left|\A_i^\mathsf{T}\y_k\right|\leq \|\A_i\|_\infty\|\y_k\|_1&\leq \max_{i\in\{1,2,\dots,n\}}\|\A_i\|_\infty=\max_{l\in\{1,2,\dots,m\}}\|\mathbf{a}_l\|_\infty<\lambda.
\end{aligned}
\end{equation*}
Therefore, for all $k\geq 1$ and $i\in\{1,2,\dots, n\}$, we have
\begin{equation}\label{xki}
\begin{aligned}
|\x_k(i)|&\geq\min\left\{\left|\x_1(i)\right|,1\right\}\geq \min\left\{\frac{\lambda}{\tau_0},1\right\},
\end{aligned}
\end{equation}
where the last inequality holds since $$|\x_1(i)|=\min\left\{\left|\x_0(i)-\frac{\A_i^\mathsf{T}\y_0}{\tau_0}\right|+\frac{\lambda}{\tau_0},1\right\}\geq \min\left\{\frac{\lambda}{\tau_0},1\right\}.$$
It follows from \eqref{xki}  that the sequence $\{|\x_k(i)|\}_k$ is bounded away from zero.
Let $\hat{\x}$ be any limit point of $\{\x_k\}$. With the selected parameters and  according to Theorem \ref{convergethe}, $\hat{\x}$ is a stationary point  of (P$_\lambda$). Lemma \ref{criticalthe} further implies that each element of $\hat{\x}$ is either $\pm 1$ or $0$. Since $\{|\x_k(i)|\}_k$ is bounded away from zero,   the elements of $\hat{\x}$ can only be $\pm1$.

\section*{Acknowledgment}

The authors would like to thank Dr. Ang Li from Xi'an Jiaotong University for his kind help on numerical simulations {\color{black} and the anonymous reviewers for their helpful comments on the paper.}

\ifCLASSOPTIONcaptionsoff
\newpage
\fi





\begin{thebibliography}{10}
\providecommand{\url}[1]{#1}
\csname url@samestyle\endcsname
\providecommand{\newblock}{\relax}
\providecommand{\bibinfo}[2]{#2}
\providecommand{\BIBentrySTDinterwordspacing}{\spaceskip=0pt\relax}
\providecommand{\BIBentryALTinterwordstretchfactor}{4}
\providecommand{\BIBentryALTinterwordspacing}{\spaceskip=\fontdimen2\font plus
\BIBentryALTinterwordstretchfactor\fontdimen3\font minus
  \fontdimen4\font\relax}
\providecommand{\BIBforeignlanguage}[2]{{%
\expandafter\ifx\csname l@#1\endcsname\relax
\typeout{** WARNING: IEEEtran.bst: No hyphenation pattern has been}%
\typeout{** loaded for the language `#1'. Using the pattern for}%
\typeout{** the default language instead.}%
\else
\language=\csname l@#1\endcsname
\fi
#2}}
\providecommand{\BIBdecl}{\relax}
\BIBdecl

\bibitem{conference}
Z.~Wu, B.~Jiang, Y.-F. Liu, and Y.-H. Dai, ``A novel negative $\ell_1$ penalty
  approach for multiuser one-bit massive {MIMO} downlink with {PSK}
  signaling,'' in \emph{Proc. IEEE Int. Conf. Acoust., Speech, Signal
  Process.}, May 2022, pp. 5323--5327.

\bibitem{massivemimo2}
F.~Rusek, D.~Persson, B.~K. Lau, E.~G. Larsson, T.~L. Marzetta, O.~Edfors, and
  F.~Tufvesson, ``Scaling up {MIMO}: Opportunities and challenges with very
  large arrays,'' \emph{IEEE Signal Process. Mag.}, vol.~30, no.~1, pp. 40--60,
  Jan. 2013.

\bibitem{massivemimo1}
J.~G. Andrews, S.~Buzzi, W.~Choi, S.~V. Hanly, A.~Lozano, A.~C.~K. Soong, and
  J.~C. Zhang, ``What will {5G} be?'' \emph{IEEE J. Sel. Areas Commun.},
  vol.~32, no.~6, pp. 1065--1082, Jun. 2014.

\bibitem{massivemimo3}
L.~Lu, G.~Y. Li, A.~L. Swindlehurst, A.~Ashikhmin, and R.~Zhang, ``An overview
  of massive {MIMO}: Benefits and challenges,'' \emph{IEEE J. Sel. Topics
  Signal Process.}, vol.~8, no.~5, pp. 742--758, Oct. 2014.

\bibitem{AD1}
O.~E. Ayach, S.~Rajagopal, S.~Abu-Surra, Z.~Pi, and R.~W. Heath, ``Spatially
  sparse precoding in millimeter wave {MIMO} systems,'' \emph{IEEE Trans.
  Wireless Commun.}, vol.~13, no.~3, pp. 1499--1513, Mar. 2014.

\bibitem{AD2}
S.~Han, C.-l. I, Z.~Xu, and C.~Rowell, ``Large-scale antenna systems with
  hybrid analog and digital beamforming for millimeter wave {5G},'' \emph{IEEE
  Commun. Mag.}, vol.~53, no.~1, pp. 186--194, Jan. 2015.

\bibitem{AD3}
F.~Sohrabi and W.~Yu, ``Hybrid digital and analog beamforming design for
  large-scale antenna arrays,'' \emph{IEEE J. Sel. Topics Signal Process.},
  vol.~10, no.~3, pp. 501--513, Apr. 2016.

\bibitem{SQUID}
S.~Jacobsson, G.~Durisi, M.~Coldrey, T.~Goldstein, and C.~Studer, ``Quantized
  precoding for massive {MU-MIMO},'' \emph{IEEE Trans. Commun.}, vol.~65,
  no.~11, pp. 4670--4684, Nov. 2017.

\bibitem{ADC4}
J.~Choi, J.~Mo, and R.~W. Heath, ``Near maximum-likelihood detector and channel
  estimator for uplink multiuser massive {MIMO} systems with one-bit {ADCs},''
  \emph{IEEE Trans. Commun.}, vol.~64, no.~5, pp. 2005--2018, May 2016.

\bibitem{ADC5}
Y.~Li, C.~Tao, G.~Seco-Granados, A.~Mezghani, A.~L. Swindlehurst, and L.~Liu,
  ``Channel estimation and performance analysis of one-bit massive {MIMO}
  systems,'' \emph{IEEE Trans. Signal Process.}, vol.~65, no.~15, pp.
  4075--4089, Aug. 2017.

\bibitem{linear2}
A.~Mezghani, R.~Ghiat, and J.~A. Nossek, ``Transmit processing with low
  resolution {D/A}-converters,'' in \emph{Proc. 16th IEEE Int. Conf. Electron.,
  Circuits Syst.}, Dec. 2009, pp. 683--686.

\bibitem{linear1}
A.~K. Saxena, I.~Fijalkow, and A.~L. Swindlehurst, ``Analysis of one-bit
  quantized precoding for the multiuser massive {MIMO} downlink,'' \emph{IEEE
  Trans. Signal Process.}, vol.~65, no.~17, pp. 4624--4634, Sept. 2017.
  \bibitem{TITpaper}
Z.~Wu, J.~Ma, Y.-F. Liu, and A.~Lee~Swindlehurst, ``Asymptotic {SEP} analysis
  and optimization of linear-quantized precoding in massive {MIMO} systems,''
  \emph{IEEE Trans. Inf. Theory (accepted)}, Oct. 2023.
\bibitem{C3PO2}
O.~Castañeda, S.~Jacobsson, G.~Durisi, M.~Coldrey, T.~Goldstein, and
  C.~Studer, ``1-bit massive {MU-MIMO} precoding in {VLSI},'' \emph{IEEE J.
  Emerg. Sel. Topics Circuits Syst.}, vol.~7, no.~4, pp. 508--522, Dec. 2017.

\bibitem{MAGIQ}
A.~Nedelcu, F.~Steiner, M.~Staudacher, G.~Kramer, W.~Zirwas, R.~S. Ganesan,
  P.~Baracca, and S.~Wesemann, ``Quantized precoding for multi-antenna downlink
  channels with {MAGIQ},'' in \emph{Proc. ITG Workshop Smart Antennas}, Mar.
  2018, pp. 1--8.

\bibitem{IDE}
C.-J. Wang, C.-K. Wen, S.~Jin, and S.-H. Tsai, ``Finite-alphabet precoding for
  massive {MU-MIMO} with low-resolution {DACs},'' \emph{IEEE Trans. Wireless
  Commun.}, vol.~17, no.~7, pp. 4706--4720, Jul. 2018.

\bibitem{ADMM}
L.~Chu, F.~Wen, L.~Li, and R.~Qiu, ``Efficient nonlinear precoding for massive
  {MIMO} downlink systems with 1-bit {DAC}s,'' \emph{IEEE Trans. Wireless
  Commun.}, vol.~18, no.~9, pp. 4213--4224, Sept. 2019.

\bibitem{CIfirst}
H.~Jedda, A.~Mezghani, J.~A. Nossek, and A.~L. Swindlehurst, ``Massive {MIMO}
  downlink 1-bit precoding with linear programming for {PSK} signaling,'' in
  \emph{Proc. IEEE Workshop Signal Process. Adv. Wireless Commun.}, Jul. 2017,
  pp. 1--5.

\bibitem{QCECI}
H.~Jedda, A.~Mezghani, A.~L. Swindlehurst, and J.~A. Nossek, ``Quantized
  constant envelope precoding with {PSK} and {QAM} signaling,'' \emph{IEEE
  Trans. Wireless Commun.}, vol.~17, no.~12, pp. 8022--8034, Dec. 2018.

\bibitem{CImodel}
A.~Li, C.~Masouros, F.~Liu, and A.~L. Swindlehurst, ``Massive {MIMO} 1-bit
  {DAC} transmission: A low-complexity symbol scaling approach,'' \emph{IEEE
  Trans. Wireless Commun.}, vol.~17, no.~11, pp. 7559--7575, Nov. 2018.

\bibitem{PBB}
A.~Li, F.~Liu, C.~Masouros, Y.~Li, and B.~Vucetic, ``Interference exploitation
  1-bit massive {MIMO} precoding: A partial branch-and-bound solution with
  near-optimal performance,'' \emph{IEEE Trans. Wireless Commun.}, vol.~19,
  no.~5, pp. 3474--3489, May 2020.

\bibitem{GMSM}
F.~Askerbeyli, W.~Xu, and J.~A. Nossek, ``1-bit precoding for massive {MIMO}
  downlink with linear programming and a greedy algorithm extension,'' in
  \emph{Proc. IEEE 93rd Veh. Technol. Conf.}, Apr. 2021, pp. 1--5.

\bibitem{CIBB}
L.~T.~N. Landau and R.~C. de~Lamare, ``Branch-and-bound precoding for multiuser
  {MIMO} systems with 1-bit quantization,'' \emph{IEEE Wireless Commun. Lett.},
  vol.~6, no.~6, pp. 770--773, Dec. 2017.

\bibitem{CI1bitoverview}
A.~Li, C.~Masouros, A.~L. Swindlehurst, and W.~Yu, ``1-bit massive {MIMO}
  transmission: Embracing interference with symbol-level precoding,''
  \emph{IEEE Commun. Mag.}, vol.~59, no.~5, pp. 121--127, May 2021.

\bibitem{sep2}
F.~Sohrabi, Y.-F. Liu, and W.~Yu, ``One-bit precoding and constellation range
  design for massive {MIMO} with {QAM} signaling,'' \emph{IEEE J. Sel. Topics
  Signal Process.}, vol.~12, no.~3, pp. 557--570, Jun. 2018.

\bibitem{sep3}
M.~Shao, Q.~Li, Y.~Liu, and W.-K. Ma, ``Multiuser one-bit massive {MIMO}
  precoding under {MPSK} signaling,'' in \emph{Proc. IEEE Global Conf. Signal
  Inf. Process.}, Nov. 2018, pp. 833--837.

\bibitem{GEMM}
M.~Shao, Q.~Li, W.-K. Ma, and A.~M.-C. So, ``A framework for one-bit and
  constant-envelope precoding over multiuser massive {MISO} channels,''
  \emph{IEEE Trans. Signal Process.}, vol.~67, no.~20, pp. 5309--5324, Oct.
  2019.

\bibitem{seppsk}
E.~S.~P. Lopes, L.~T.~N. Landau, and A.~Mezghani, ``Minimum symbol error
  probability discrete symbol level precoding for {MU-MIMO} systems with {PSK}
  modulation,'' \emph{IEEE Trans. Commun. (accepted)}, Jul. 2023.

\bibitem{digitalcommunication}
J.~G. Proakis and M.~Salehi, \emph{Digital Communications}, 5th~ed.\hskip 1em
  plus 0.5em minus 0.4em\relax New York, NY, USA: McGraw-Hill, 2008.

\bibitem{CItutorial}
A.~Li, D.~Spano, J.~Krivochiza, S.~Domouchtsidis, C.~G. Tsinos, C.~Masouros,
  S.~Chatzinotas, Y.~Li, B.~Vucetic, and B.~Ottersten, ``A tutorial on
  interference exploitation via symbol-level precoding: Overview,
  state-of-the-art and future directions,'' \emph{IEEE Commun. Surveys Tuts.},
  vol.~22, no.~2, pp. 796--839, 2nd Quart. 2020.

\bibitem{diversity}
Z.~Wu, J.~Wu, W.-K. Chen, and Y.-F. Liu, ``Diversity order analysis for
  quantized constant envelope transmission,'' \emph{IEEE Open J. Signal
  Process.}, vol.~4, pp. 21--30, Jan. 2023.

\bibitem{CI1}
C.~Masouros, T.~Ratnarajah, M.~Sellathurai, C.~B. Papadias, and A.~K. Shukla,
  ``Known interference in the cellular downlink: a performance limiting factor
  or a source of green signal power?'' \emph{IEEE Commun. Mag.}, vol.~51,
  no.~10, pp. 162--171, Oct. 2013.

\bibitem{CI2}
C.~Masouros, M.~Sellathurai, and T.~Ratnarajah, ``Vector perturbation based on
  symbol scaling for limited feedback {MISO} downlinks,'' \emph{IEEE Trans.
  Signal Process.}, vol.~62, no.~3, pp. 562--571, Feb. 2014.

\bibitem{CI3}
C.~Masouros and G.~Zheng, ``Exploiting known interference as green signal power
  for downlink beamforming optimization,'' \emph{IEEE Trans. Signal Process.},
  vol.~63, no.~14, pp. 3628--3640, Jul. 2015.

\bibitem{CIequivalent}
A.~Li, C.~Masouros, B.~Vucetic, Y.~Li, and A.~L. Swindlehurst, ``Interference
  exploitation precoding for multi-level modulations: Closed-form solutions,''
  \emph{IEEE Trans. Commun.}, vol.~69, no.~1, pp. 291--308, Jan. 2021.

\bibitem{ICASSPQCE}
Z.~Wu, Y.-F. Liu, B.~Jiang, and Y.-H. Dai, ``Efficient quantized constant
  envelope precoding for multiuser downlink massive {MIMO} systems,'' in
  \emph{Proc. IEEE Int. Conf. Acoust., Speech, Signal Process.}, Jun. 2023.

\bibitem{techreport}
\BIBentryALTinterwordspacing
Z.~Wu, B.~Jiang, Y.-F. Liu, M.~Shao, and Y.-H. Dai, ``{E}fficient {CI}-based
  one-bit precoding for multiuser downlink massive {MIMO} systems with {PSK}
  modulation,'' 2023. [Online]. Available:
  \url{https://arxiv.org/abs/2110.11628}
\BIBentrySTDinterwordspacing

\bibitem{homotopy3}
L.~Xiao and T.~Zhang, ``A proximal-gradient homotopy method for the sparse
  least-squares problem,'' \emph{SIAM J. Optim.}, vol.~23, no.~2, pp.
  1062--1091, 2013.

\bibitem{homotopy1}
M.~Shao and W.-K. Ma, ``Binary {MIMO} detection via homotopy optimization and
  its deep adaptation,'' \emph{IEEE Trans. Signal Process.}, vol.~69, pp.
  781--796, Feb. 2021.

\bibitem{penaltyproof}
B.~Jiang, Y.-F. Liu, and Z.~Wen, ``${L}_p$-norm regularization algorithms for
  optimization over permutation matrices,'' \emph{SIAM J. Optim.}, vol.~26,
  no.~4, pp. 2284--2313, 2016.

\bibitem{OFDM}
F.~Askerbeyli, H.~Jedda, and J.~A. Nossek, ``1-bit precoding in massive
  {MU-MISO-OFDM} downlink with linear programming,'' in \emph{Proc. ITG
  Workshop Smart Antennas}, Apr. 2019.

\bibitem{HiBSA}
S.~Lu, I.~Tsaknakis, M.~Hong, and Y.~Chen, ``Hybrid block successive
  approximation for one-sided non-convex min-max problems: Algorithms and
  applications,'' \emph{IEEE Trans. Signal Process.}, vol.~68, pp. 3676--3691,
  Apr. 2020.

\bibitem{xu2020unified}
Z.~Xu, H.~Zhang, Y.~Xu, and G.~Lan, ``A unified single-loop alternating
  gradient projection algorithm for nonconvex-concave and convex-nonconcave
  minimax problems,'' \emph{Math. Program.}, vol. 201, no.~1, pp. 635--706,
  2023.
  

\bibitem{MMPD}
Y.~Chen, G.~Lan, and Y.~Ouyang, ``Optimal primal-dual methods for a class of
  saddle point problems,'' \emph{SIAM J. Optim.}, vol.~24, no.~4, pp.
  1779--1814, 2014.

\bibitem{zhang2020singleloop}
J.~Zhang, P.~Xiao, R.~Sun, and Z.-Q. Luo, ``A single-loop smoothed gradient
  descent-ascent algorithm for nonconvex-concave min-max problems,'' in
  \emph{NeurIPS}, vol.~33, Dec. 2020, pp. 7377--7389.

\bibitem{pmlr-v119-lin20a}
T.~Lin, C.~Jin, and M.~Jordan, ``On gradient descent ascent for
  nonconvex-concave minimax problems,'' in \emph{Proc. Int. Conf. Mach. Learn.
  Cybern.}, vol. 119, Jul. 2020, pp. 6083--6093.

\bibitem{dai2020majorized}
\BIBentryALTinterwordspacing
Y.-H. Dai, J.~Wang, and L.~Zhang, ``Majorized semi-proximal alternating
  coordinate method for nonsmooth convex-concave minimax optimization,'' 2020.
  [Online]. Available: \url{https://arxiv.org/abs/2012.11502}
\BIBentrySTDinterwordspacing

\bibitem{proximal}
N.~Parikh and S.~Boyd, ``Proximal algorithms,'' \emph{Found. Trends Optim.},
  vol.~1, no.~3, pp. 127--239, 2014.

\bibitem{projection2}
M.~Held, P.~Wolfe, and H.~P. Crowder, ``Validation of subgradient
  optimization,'' \emph{Math. Program.}, vol.~6, no.~1, pp. 62--88, 1974.

\bibitem{projection}
L.~Condat, ``Fast projection onto the simplex and the $l_1$ ball,'' \emph{Math.
  Program.}, vol. 158, no.~1, pp. 575--585, 2016.


\bibitem{np}
M.~R. Garey and D.~S. Johnson, \emph{Computers and Intractability: A Guide to
  the Theory of NP-Completeness}. San Francisco, CA, USA:  W. H. Freeman, 1979.

\bibitem{convexanalysis} 
R.~T. Rockafellar, \emph{Convex Analysis}. Princeton, NJ, USA: Princeton Univ. Press, 1970.




\end{thebibliography}
%
\vspace{-0.2cm}
\end{document}